%% file: Qcap.tex
\documentclass[11pt,reqno]{amsart}
\usepackage{mathrsfs,amsmath,amsxtra,amssymb,amsthm,amsfonts}%,showkeys}
\usepackage{graphicx}
\usepackage{amsaddr}
\usepackage{float}
\usepackage{authblk}
\usepackage[T1]{fontenc}
\usepackage[svgnames]{xcolor}
\numberwithin{equation}{section}
\newcommand{\Z}{\mathbb{Z}}

\def\fr{\begin{align*}}

\newcommand{\kl}{\pl \le \pl}
\newcommand{\gl}{\pl \ge \pl}

\newcommand{\lel}{\pl = \pl}

\newcommand{\nz}{{\mathbb N}}
\newcommand{\nen}{n \in \nz}

\newcommand{\cz}{{\mathbb C}}

\newcommand{\ten}{\otimes}

\newcommand{\pl}{\hspace{.1cm}}
\newcommand{\pll}{\hspace{.3cm}}

\newcommand{\ran}{\rangle}
\newcommand{\lan}{\langle}

\newcommand{\om}{\omega}

\newcommand{\al}{\alpha}
\newcommand{\si}{\sigma}

\newcommand{\la}{\lambda}

\newcommand{\F}{{\mathcal F}}
\newcommand{\E}{{\mathcal E}}

\newcommand{\D}{{\mathcal D}}

\newcommand{\N}{{\mathcal N}}

\newcommand{\fs}[2]{\frac{#1}{#2}}

\newcommand{\norm}[2]{\parallel \! #1 \! \parallel_{#2}}
\newcommand{\dep}[1]{\frac{d}{dp}#1|_{p=1}}
\renewcommand{\thefootnote}{\fnsymbol{footnote}}
\newcommand{\C}{\mathcal{C}}
\newtheorem{lemma}{Lemma}[section]
\newtheorem{prop}[lemma]{Proposition}
\newtheorem{theorem}[lemma]{Theorem}
\newtheorem{cor}[lemma]{Corollary}

\newtheorem{rem}[lemma]{Remark}

\newcommand{\re}{\begin{rem}\rm}
\newcommand{\mar}{\end{rem}}
\newtheorem{exam}[lemma]{Example}
\newcommand{\bra}[1]{\langle{#1}|}
\newcommand{\ket}[1]{|{#1}\rangle}
\newcommand{\ketbra}[1]{|{#1}\rangle\langle{#1}|}
\newcommand{\qd}{\end{proof}\vspace{0.5ex}}
\newcommand{\prf}{\begin{proof}[\bf Proof:]}

\newcommand{\xspace}{\hbox{\kern-2.5pt}}

\newcommand{\rra}{\rightarrow }

\makeatletter
\renewcommand\@biblabel[1]{#1.}
\makeatother

\DeclareMathOperator{\flip}{flip}
\DeclareMathOperator{\st}{st}
\DeclareMathOperator{\crp}{rc}

\allowdisplaybreaks
\newtheorem{defi}[lemma]{Definition}
\makeatletter
\g@addto@macro{\endabstract}{\@setabstract}
\newcommand{\authorfootnotes}{\renewcommand\thefootnote{\@fnsymbol\c@footnote}}%
\makeatother
\begin{document}

\begin{center}
  \LARGE
  Capacity Bounds via Operator Space Methods \par \bigskip

  \normalsize
  \authorfootnotes
  Li Gao\textsuperscript{1}, Marius Junge\textsuperscript{1}$^{,}$\footnote[1]{MJ is partially supported by {NSF-DMS} 1501103.}, and
  Nicholas LaRacuente\textsuperscript{2}$^{,}$\footnote[2]{NL is supported by NSF Graduate Research Fellowship Program DGE-1144245.}\par \bigskip

  \textsuperscript{1}Department of Mathematics, University of Illinois, Urbana, IL 61801, USA \par
  \textsuperscript{2}Department of Physics, University of Illinois, Urbana, IL 61801, USA\par \bigskip
\end{center}

\input{Qcap_abstract_nick}
\input{Qcap_introduction_nick}
\input{Qcap_preliminaries}
\input{Qcap_mainresults}
\input{Qcap_operatorspace}

\input{Qcap_stinespringspaces}

\input{Qcap_comparisonthm}
\input{Qcap_cbentropy}

\input{Qcap_examples}
\input{Qcap_bib}

%\tableofcontents
\end{document}

%% file: Qcap_abstract_nick.tex
\renewcommand{\abstractname}{\bf{ABSTRACT}}
\begin{abstract}
Quantum capacity, the ultimate transmission rate of quantum communication, is characterized by regularized coherent information. In this work, we reformulate approximations of the quantum capacity  by operator space norms and give both upper and lower estimates on quantum capacity and potential quantum capacity using complex interpolation techniques from operator space theory. Upper bounds are obtained by a comparison inequality for R{\'e}nyi entropies. Analyzing the maximally entangled state for the whole system and for error-free subsystems provides lower bounds for the ``one-shot'' quantum capacity. These two results combined give upper and lower bounds on quantum capacity for our ``nice'' classes of channels, which differ only up to a factor $2$, independent of the dimension. The estimates are discussed for certain classes of channels, including group channels, generalized Pauli channels and other high-dimensional channels.\\
%{\bf{Key words:} Quantum Capacity, cb-entropy, Interpolation, Operator Space, Conditional Expectation}
\end{abstract} 

%% file: Qcap_introduction_nick.tex
\section{Introduction}
\noindent The aim of quantum Shannon theory is to extend Shannon's information theory, formulated in his landmark paper \cite{48}, and provide the proper framework in the context of quantum mechanics, including non-locality \cite{bell,EPR}. In recent decades, vast progress has been made in extending Shannon's theory for quantum channels and their capacities. Moreover, the role  of different resources such as entanglement, transmission of classical and quantum  bits and their interaction has significantly improved (see e.g. \cite{mother,Family,DS}).
A surprising but important feature in quantum Shannon theory is the variety of capacities associated with a quantum channel. For instance, the \emph{classical capacity} \cite{Holevo,SW} describes the capability of classical information transmission through a quantum channel; \emph{entanglement-assisted classical capacity} \cite{BSST} considers classical transmission using additional entanglement accessible to the sender Alice and the receiver Bob. One big success in quantum information theory is the quantum capacity theorem proved
by Lloyd \cite{Lloyd}, Shor \cite{Shor} and Devetak \cite{Devetak} with increasing standards of rigor. It demonstrates that the \emph{quantum capacity} $Q(\Phi)$ of a channel $\Phi$, as the ultimate capability of $\Phi$ to transmit quantum information, is characterized by the \emph{regularized coherent information} as follows:
\begin{align}\label{qcapacity}
 Q(\Phi)\lel\lim_{k\to \infty} \frac{Q^{(1)}(\Phi^{\ten k})}{k}\pl,\pl
Q^{(1)}(\Phi)\lel\max_{\rho \pl  \text{\tiny pure}}I_c(A\ran B)_\si\pl,
\end{align}
where $\si^{AB}=id_A\ten \Phi(\rho^{AA'})$ and the maximum runs over all pure bipartite state $\rho^{AA'}$. $I_c(A\ran B)_\si$ is the coherent information of bipartite $\si$ given by $H(\si^B)-H(\si^{AB})$, with $H(\si)=-tr(\si\log \si)$ being the von Neumann entropy, and $Q^{(1)}$ is the ``one-shot'' quantum capacity. Let us also recall that the negative cb-entropy (also called the reverse coherent information) of a channel $\Phi$ is defined similarly as
${-S_{cb}(\Phi)=\max_{\rho} H(A)_\rho-H(AB)_\rho} $ (see Section $2$ for formal definitions).

Despite of this impressive theoretical success, there are few classes of quantum channels which have a closed, computable formula for the quantum capacity. The mathematical reason is the necessity to consider the limit in   \eqref{qcapacity},
the so-called \emph{regularization}, which amounts to making calculations for channels with arbitrary large inputs and outputs. It is known that for qubit depolarizing channels the regularization is strictly greater than the ``one-shot'' expression \cite{sa, sa1}. Moreover it was proved in \cite{cubitt} that for any $k\in \mathbb{N}$, there exists a channel $\Phi$ such that
the regularization of $k$ uses of $\Phi$ is one, but adding one more copy makes it positive, i.e.  ${Q^{(1)}(\Phi^{\ten (k+1)})>Q^{(1)}(\Phi^{\ten k})=0}$. As of today, calculation of quantum capacities is possible only for specific channels \cite{erasure,Neufang,adephasing}. Devetak and Shor in \cite{DS} proved that $Q=Q^{(1)}$ for degradable channels, those for which the environment can be retrieved from Bob's output with the help of another channel. Hence regularization is not necessary for degradable channels. For non-degradable channels, little is known about the exact value of quantum capacity. Several different methods have been introduced to give estimates on particular or general channels \cite{Holevo2,Winterss,approximate,sc,XW}.

The aim of this work is to introduce complex interpolation techniques to estimate the quantum capacity $Q$ from above and below for large, nice  classes of channels. The upper and lower bounds only differ by a factor of $2$. These in general non-degradable channels can be viewed as perturbations of the so-called conditional expectations, projections onto $C^*$-subalgebras. In finite dimensions, conditional expectations are direct sums of partial traces, hence they have clear capacity formula by observations of Fukuda and Wolf in \cite{Wolf}. Based on that, we observe a ``comparison property'' on entropy and capacity on our nice class of channels. Related estimates for the potential quantum capacity and the quantum dynamic capacity region also follow from the ``comparison property''.  Moreover, with similar assumptions we prove a formula for the negative cb-entropy.

Here we briefly formulate our results for certain random unitary channels which fall in our nice class. Let $G$ be a finite group of  order $|G|=n$ and the left regular representation given by $\la(g)(e_h)=e_{gh}$ on Hilbert space $\ell_2(G)\cong l_2^n$. Here $e_g(h)=\delta_{g,h}$ are the standard unit vectors for $\ell_2(G)$. There is also a right regular representation ${r(g)(e_h)=e_{hg^{-1}}}$. The group von Neumann algebra is $L(G)={\rm span}\{\la(g)|g\in G\}$ with commutant ${L(G)'=\{T| \pl\forall {x\in L(G)} \pl,  Tx=xT\}}$ given by the right regular representation ${L(G)'=R(G)={\rm span}\{r(g)|g\in G\}}$ (see e.g. \cite{Tak}). Given a function $f:G\to \cz$ with $f(g)\gl 0$ and $\sum_g f(g)=n$, we may define the channel
\begin{align}\label{group}
\theta_f(\rho) \lel \frac1n\sum_g f(g) \la(g)\rho \la(g)^* \pl .
\end{align}
In general, such a random unitary channel is not degradable unless $G$ is abelian. $L(G)$ is a finite dimensional $C^*$-algebra and hence admits a decomposition $L(G)=\oplus_k M_{n_k}$ into matrix blocks, given by a complete list of irreducible representations. We obtain the following estimates for the quantum capacity:
\begin{theorem}\label{group01} Let $G$ be a finite group such that ${L(G)=\oplus_k  M_{n_k}}$, and
$\theta_f$ defined as above. Then
 \begin{align} \label{com}\max\{\log (\max_k n_k), -S_{cb}(\theta_f)\}
&\kl Q^{(1)}(\theta_f)\kl Q(\theta_f)\kl \log (\max_k n_k)+ (-S_{cb}(\theta_f)) \pl, \\
 -S_{cb}(\theta_f) &= \log n- H(\frac{1}{n}f) \label{cb} .\end{align}
\end{theorem}
\noindent Here $H(\frac{1}{n}f)=-\sum_g \frac{f(g)}{n}\log \frac{f(g)}{n}$ is the Shannon entropy. The formula for the cb-entropy of (quantum) group channels has been discovered in the unpublished paper \cite{JNR2} (reproved here), a common source of inspiration for this work and \cite{Neufang}.  The upper bound tackles, up  to a factor 2, the problem of regularization for this class of non-degradable channels. Our results are particularly striking for non-abelian $G$ with $\max_k n_k \ll |G|^{1/2}$. Additionally, Theorem \eqref{group01} holds verbatim for quantum groups. We have two motivations for considering quantum groups. First, quantum groups provide new examples of channels with Kraus operators which are neither unitaries nor projections. Second, some variations of quantum group operations relate to Kitaev's work \cite{kitaev} on anyons. %In order to illustrate the method we recall that the estimate $\max_k n_k\le \sqrt{n}$, means that the dimension of the largest irreducible representation of a finite group can not exceed the square root of $|G|=n$. Similarly for a quantum group $\A=\oplus_k M_{n_k}$, a $C^*$-algebra which admits a co-product and an antipode, we derive
 %\[ \max_{\pi \mbox{  \scriptsize{irreducible}}}   \dim(\pi)
 %\lel 2^{Q(\theta_1)} \kl 2^{Q_{EA}(\theta_1)}=  \sqrt{\dim(A)} \]
%since entanglement always increases the capacity. We suspect that  more representation theory of group and quantum group is required for analyzing capacities of quantum group channels given by co-representations.
It appears that there is an interesting link between representation theory and capacity.

Our proof relies heavily on operator space tools, in particular complex interpolation. The connection between operator spaces and quantum information has long been noted. In particular,  the additivity of the cb-entropy can be derived by differentiating completely bounded norms \cite{DJKR}. In \cite{GW} Gupta and Wilde used the same completely bounded norm to prove the strong converse of entanglement-assisted classical capacity. Junge and Palazuelos found a reformulation of entanglement-assisted classical capacity and Holevo capacity in terms of the completely $p$-summing norm \cite{psumming}. Based on this, they also gave a super-additivity example of $d$-restricted entanglement-assisted classical capacity \cite{JP}. Our work discovers connections between quantum capacity and operator space structures and introduce interpolation technique to estimate the R{\'e}nyi entropy and information measures.% The operator space approach to quantum capacity is based on the fact that Shannon's entropy is the limiting case of the Renyi entropies (for $p>1$),
%Hence it is in spirit related to previous work on Renyi entropies (see e.g. \cite{WWY,renyi}).

We organize this work as follows. The next section reviews basic definitions about channels and capacities. In Section 3, we state our main theorem and derive our upper bounds based on the ``comparison property''
 \[  \|(id\ten \theta_1)(\rho)\|_p\kl   \|(id\ten \theta_f)(\rho)\|_p\kl \|f\|_p  \|(id\ten \theta_1)(\rho)\|_p \pl, \]
 where $\norm{\cdot}{p}$ denotes the Schatten-$p$ norm. This section provide the basic idea of our estimates, postponing operator space terminology and proof. In Section 4 we deliver basic operator space and interpolation theory necessary for the rest of the paper. Section 5 introduces the Stinespring space of a channel and its connection to quantum capacity. %Using this, we calculate $Q^{(1)}$ and potential quantum capacity of conditional expectations.
 Section 6 is devoted to the proof of the ``comparison property''. Section 7 discusses cb-entropy and combined upper and lower bounds. Section 8 provides six examples including the group channels we see above. 

%% file: Qcap_preliminaries.tex
\section{Preliminaries}
\subsection{States and channels} We denote by $B(H)$ the space of bounded operators on Hilbert space $H$. In this paper, we restrict oursevles to finite dimensional Hilbert spaces and write $\displaystyle \dim H=|H|.$ Sometimes we also use the matrix algebra $M_n\cong B(l^n_2)$ where $l^n_2$ is the standard $n$-dimensional Hilbert space. For $1\le p< \infty$, the Schatten-$p$ norm of an operator $a\in B(H)$ is defined as
\[\norm{a}{p}=tr((a^*a)^{\frac{p}{2}})^{\frac1p}\pl ,\]
where ``$tr$'' is the standard trace on matrix algebra. In particular, $p=\infty$ denotes the usual operator norm, and $p=1$ is called the trace class norm. We denote $S_p(H)$ (or $S^n_p$) as the Banach space $B(H)$ (respectively $M_n$) equipped with the Schatten-$p$ norm. A \emph{state} of the system of Hilbert space $H$ is given by a density operator $\rho\in B(H)$, i.e.  $\rho\ge 0,\ {tr(\rho)=1}$. Following the duality between the Schr\"odinger and Heisenberg pictures, we view the density $\rho$ as an element in the trace class operators $S_1(H)$, which is the Banach space pre-dual of $B(H)$. A state is called \emph{pure} if its density is a rank one projector. Pure states are extreme points of the set of states. The identity operator in $B(H)$ is denoted as $1$ and $\frac{1}{|H|}1$ as a density operator is called the \emph{totally mixed state}.

We index physical systems by capital letters and the corresponding Hilbert spaces by subscripts. For example, it is common to assume Alice is in hold of system $H_{A'}$ and Bob $H_B$, whereas $H_A$ and $H_E$ are the reference system and environment respectively. The bipartite system is denoted as $H_{AB}\cong H_A\ten H_B$. For a multipartite state, we use the superscripts to track the systems of the states, i.e. for a state $\rho^{AB}\in S_1(H_{AB})$,
$\rho^{A}=id_A \ten tr_B (\rho^{AB})$ is the reduced density operator on $A$. Here $id_A$ is the identity map on $B(H_A)$ whereas the identity operator in $B(H_A)$ will be denoted by $1_A$, and $tr_B$ is the trace on $B(H_B)$. %A bipartite state $\rho^{AB}$ is separable if it is a convex combination of product states, i.e. $\rho^{AB}=\sum_i p(i)\rho_i^A\ten \rho_i^B$ for some states $\rho_i^A$ and $\rho_i^B$. A state is entangled if it is not separable.
A pure bipartite state of unit vector $\ket{\psi}^{AA'}$ is a \emph{maximally entangled state} if $\ket{\psi}=\frac{1}{|H_A|}\sum_i e^A_i\ten e^{A'}_i$ with two orthogonal bases $\{e^A_i\}$ and $\{e^{A'}_i\}$.

A quantum channel from Alice to Bob is mathematically a completely positive and trace preserving (CPTP) map $\displaystyle {\Phi:S_1(H_{A'})\to S_1(H_B)}$, i.e. $id_A\ten \Phi (\rho^{AA'})$ is again a state in $S_1(H_{AB})$ for all bipartite states $\rho^{AA'}\in S_1(H_{AA'})$ with any reference systems $H_A$. Two equivalent definitions of quantum channels will also be used:
\begin{enumerate}
  \item[i)] Kraus operators: there exists a finite sequence of operators $x_i\in B(H_{A'}, H_B)$ satisfying $\sum_i x_i^*x_i=1_{A'}$, s.t. $\displaystyle \Phi(\rho)=\sum_i x_i\rho x_i^*$;
  \item[ii)] Stinespring dilation: there exists an environment Hilbert space $H_E$ and a partial isometry ${V\in B(H_{A'}, H_B\ten H_E)}$ with $V^*V=1_{{A'}}$, s.t.
      \begin{align} \label{dilation}\Phi(\rho)=id_B\ten tr_E(V\rho V^*)\pl .\end{align}
\end{enumerate}
The Stinespring dilation leads to the complementary channel of $\Phi$:
\[\Phi^E(\rho)=tr_B\ten id_E(V\rho V^*)\pl,\]
for which the outputs are sent to the environment. A channel $\Phi$ is \emph{degradable} if there exists another channel $\Psi$ such that $\Phi^E=\Psi\circ\Phi$.
A well-studied class of degradable channels are Hadamard channels, which have a general form as following:
\[\Phi(\rho)=\sum_{1\le i,j\le n}\bra{h_i} \rho\ket{h_j}\lan k_i|k_j\ran e_{i,j}\pl ,\]
where $\sum_{i\le n}\ketbra{h_i}=1$, $\ket{k_i}$'s are unit vectors and $e_{i,j}$'s are the matrix units. Here and in the following we use the standard bra-ket notation.\\
\subsection{Information measures} Given that $\rho$ is a density matrix, the von Neumann entropy of $\rho$ is closely related to its Schatten $p$-norms as follows,
\begin{align} H(\rho)\lel -tr(\rho \ln \rho)=\lim_{p\to 1^+} \frac{1-\|\rho\|_p}{p-1} \pl. \label{pentropy}\end{align}
As a matter of convenience, we use the natural logarithm for the definition of entropy, which differs to the logarithm with base $2$ by a constant scalar $\ln 2$. All the main results hold verbatim if the natural logarithm is replaced by $\log_2$, in the usual unit of (qu)bit. For a bipartite state $\rho^{AB}$ the mutual information $I(A:B)_{\rho}$ and the coherent information $I_c(A\rangle B)$ are defined as
\[
 I(A;B)_{\rho} :\lel H(A)_{\rho}+H(B)_{\rho}-H(AB)_{\rho}\pl\ ,\ \pl I_c(A\rangle B)_{\rho}
 :\lel H(B)_{\rho}-H(AB)_{\rho}  \pl ,\]
where $H(A)_{\rho}=H(\rho^A)$, $H(AB)=H(\rho^{AB})$.
If the state $\rho$ is clear from the context, the subindex is often omitted.
%The coherent information is a fundamental quantity in quantum information theory. Indeed, it has no classical analog, since it is always negative if $B$ is a classical system. It is of more meaning when considering the complementary channel: given $\Phi$ as in \eqref{dilation}, we consider $\ket{\psi^{AA'}}$ to be a purification of $\rho^{A'}$ and denote ${\sigma^{ABE}=id_A\ten V(\ketbra{\psi})V^*}$. Then the coherent information of $\Phi$ with input $\rho$ is
%\[I_c(A\rangle B)_\sigma= H(B)_\si-H(AB)_\si=H(B)_\si-H(E)_\si\pl .\]
%The last equality follows from the fact $\sigma^{ABE}$ is a pure state.
\subsection{Channel capacity}
Let us briefly review different quantum channel capacities which will be considered in this paper. Here we only state the rate definition of quantum capacity $Q$ but refer to \cite{Wildebook} for similar rate definitions of other capacities. Given a channel $\Phi$,
  a $(n, m,\epsilon)$-\emph{quantum code} is a pair of completely positive and trace preserving maps $(\C,\D)$,
 \[\C: S_1^m \to S_1(H_{A'}^{\ten n})\ \ \ \  \D: S_1(H_{B}^{\ten n}) \to S_1^m\pl,\]
 such that
 \[\norm{id_{m}\ten(\D\circ\N^{\ten^n}\circ \C)(\phi)-\phi}{1} \pl\le \pl \epsilon\pl ,\]
 where $\phi$ is a maximally entangled state in $S_1^m\ten S_1^m$, and $id_m$ is the identity map on $S_1^m$. The maps $\C$ and $\D$ are called the encoding and decoding respectively. A non-negative number $R$ is a \emph{achievable rate of quantum communication} if for any $\epsilon>0$ there exists an $(n, m,\epsilon)$ code such that $\displaystyle\frac{\ln m}{n}\ge R- \epsilon$. Then the \emph{quantum capacity} of $\Phi$, denoted $Q(\Phi)$, is defined as the supremum of all achievable rates $R$.

The quantum capacity theorem (also known as the LSD theorem) states that for a quantum channel $\Phi$, the capacity to transmit quantum information is \begin{align}
Q(\Phi)=\lim_{k\to \infty}\frac{Q^{(1)}(\Phi^{\ten k})}{k}\pl\ \ ,\ \
Q^{(1)}(\Phi)=\max_{\rho^{AA'}\ \text{\tiny pure}}I_c(A\ran B)_\si\pl,
\end{align}
where $\si^{AB}=id_A\ten \Phi(\rho^{AA'})$ is the output of channel. The maximum runs over all pure bipartite states $\rho^{AA'}$, and by convexity it is equivalent to consider any bipartite states. We will also be concerned with entanglement-assisted classical capacity denoted by $C_{EA}$. The entanglement-assisted classical capacity theorem \cite{BSST} shows that for a quantum channel $\Phi$, the capacity to transmit classical information with unlimited entanglement-assistance is
 \begin{align} C_{EA}(\Phi) \lel \max_{\rho_{AA'} \ {\tiny pure}}
 I(A;B)_\si \pl .
\end{align}Again the maximum runs over all pure bipartite inputs $\rho^{AA'}$. The potential capacities were introduced in \cite{WD}  by Winter and Yang to consider the maximal possible superadditivity of capacities. In this paper, we only consider the single-letter potential quantum capacity defined as follows:
 \begin{align} Q^{(p)}(\Phi) \lel \sup_{\Psi}\pl Q^{(1)}(\Phi\ten \Psi)-Q^{(1)}(\Psi)\pl, \end{align}
where the maximum runs over arbitrary channel $\Psi$. Note that we use a different notation ``$Q^{(p)}$'' from ``$V^{(1)}$'' in \cite{Winterss}, respectively ``$Q^{(1)}_p$'' in \cite{WD} to save the symbol ``$Q_p$'' for later use. By definition, we have  $Q^{(p)}\ge Q\ge Q^{(1)}$. $\Phi$ is \emph{strongly additive} on $Q^{(1)}$ if $Q^{(p)}=Q^{(1)}$, i.e. $Q^{(1)}(\Phi\ten \Psi)=Q^{(1)}(\Phi)+Q^{(1)}(\Psi)$ for arbitrary $\Psi$. Another information measure we will consider in this paper is the negative $cb$-entropy introduced in \cite{DJKR}:
\begin{align}\label{Scb} -S_{cb}(\Phi) \lel \max_{\rho^{A'A},\pl \text{pure}} H(A)_\si-H(AB)_\si\pl . \end{align}
It is also called reverse coherent information, and an operational meaning is discussed in \cite{invcoh}.

Finally, we will apply our estimates to the quantum dynamic region. Hsieh and Wilde introduced the quantum dynamic region $C_{CQE}$ to describes the resources traded off with a quantum channel \cite{HW} . ``$C$'' represents classical information transmission, ``$Q$'' represents qubit transmission and ``$E$'' is the entanglement distribution. We refer to their paper \cite{HW} and  Wilde's book \cite{Wildebook} for a formal definition of $C_{CQE}$. Here we state the quantum dynamic theorem from \cite{HW} for the convenience of readers. For a quantum channel $\Phi: S_1(H_{A'})\to S_1(H_B)$, its dynamic capacity region $C_{CQE}$ is characterized as following:
\[
C_{CQE}(\Phi)=\overline{\bigcup^{\infty}_{k=1}\frac{1}{k}C_{CQE}^{(1)}(\Phi^{\otimes k})}\ , \ \ \ \ \ C_{CQE}^{(1)}\equiv \bigcup_\sigma C^{(1)}_{CQE,\sigma}
\]
where the overbar indicates the closure of a set. The ``one-shot" region $C_{CQE}^{(1)}\subset \mathbb{R}^3$ is the union of the ``one-shot, one-state" regions $C_{CQE,\si}^{(1)}$, which are the sets of all rate triples $(C,Q,E)$ such that:
\begin{align*}
C+2Q\le I(AX;B)_\sigma\pl, \ Q+E\le I(A\rangle BX)_\sigma\pl , \  C+Q+E&\le I(X;B)_\sigma+I(A\rangle BX)_\sigma\pl.
\end{align*}
The above entropy quantities are with respect to a classical-quantum state
\[\sigma^{XAB}\lel\sum_x p_X(x)\ketbra{x}^X\otimes(id_A \ten \Phi^{A'\to B})(\rho^{AA'}_x)\]
and the states $\rho^{AA'}_x$ are pure.
\subsection{Von Neumann algebras}
Let us recall that a von Neumann algebra is a weak$^*$-closed $^*$-subalgebra of $B(H)$ for some Hilbert space $H$. We say $\tau$ is a normal faithful trace on the von Neumann algebra $N$ if $\tau:N_+\to [0,\infty]$ satisfies
\begin{enumerate}
\item[i)] $\tau(x+y)=\tau(x)+\tau(y)$;
\item[ii)] $\tau(u^*xu)=\tau(x)$ for all unitaries $u$;
\item[iii)] $\tau(x)=\sup_{0\le x\le y,\tau(y)<\infty}\tau(y)$;
 \item[iv)] $\tau(x)=0$ iff $x=0$.
\end{enumerate}
Here $x,y\in N_+=\{z^*z| z\in N\}$ is the cone of positive elements. In additional, $\tau$ is called normalized if $\tau(1)=1$. For $1\le p\le\infty$, the $L_p$-norm with respect to trace $\tau$ is defined by
\[\norm{a}{p}\lel \tau((a^*a)^{\frac{p}{2}})^{\frac1p}\pl, \ \  \  a\in N\pl,
\]
which is a generalization of Schatten-$p$ norms on $N$. A density $\rho\in N$ is a positive element with trace $\tau(\rho)=1$. In operator algebra literature, a state on $N$ is a unital positive linear functional $\phi: N \to \mathbb{C}$, and again by duality, a state is also given by a density $\rho$ in $N$, i.e. $\phi_{\rho}(T)=tr(\rho T)$. %In most of our applications $N$ is finite dimensional, hence it is equivalent to say $\rho$ is a state or a density.

For a given state $\phi$ on $N$, the GNS construction is given by the triple $(H_\phi,\pi_\phi, \xi_\phi)$. The Hilbert space $H_\phi=L_2(N,\phi)$ is the completion of $N$ with inner product $(x,y)=\phi(x^*y)$ and $\xi_\phi=\ket{1}$ is given by the corresponding vector of identity in $L_2(N,\phi)$. Then the GNS representation $\pi_{\phi}$ is $\pi_{\phi}(x)\ket{y}=\ket{xy}$. If $\phi$ is a normal faithful state, $\xi_{\phi}$ is also separating, and there exists an anti-linear isometry $J$ such that $JNJ=N'$ holds for the commutant. In our case, we call the inclusion $N\subset B(H)$ a \emph{standard inclusion} if $H\cong  L_2(N,\phi)$ for some faithful state $\phi$. See Section 5 for more information on standard inclusions. 

%% file: Qcap_mainresults.tex
\section{Capacity bounds via comparison theorem}
%\noindent In this section, we introduce channels defined by certain symbols, which allow us to state our main results. We postpone the proof of the Theorem \ref{comp}, which requires operator space and operator algebra technology. Assuming the theorem \ref{comp}, however, we will provide the proof of the upper estimate starting from $L_p$ norm estimates.

\subsection{VN-Channels}
We are interested in classes of channels indexed by densities from a von Neumann algebra. Indeed,  let $N$ be a von Neumann algebra with a faithful normalized trace $\tau$ and $U\in M_m\ten N$ be a unitary. For each density $f\in N$ we may  introduce a channel $\theta_f: S_1^m \to S_1^m$ as follows:
\begin{align}\theta_f(\rho)= id_{m}\ten\tau(U(\rho\ten f)U^*)\pl ,\label{channel3}\end{align}
Note that the map \eqref{channel3} is completely positive and trace preserving if and only if $f$ is a density. We use the normalized trace on $N$ so that that the identity operator $1$ becomes a density in $N$ (i.e. $\tau(a)=1$). Our main goal is to understand perturbations of quantum capacity on the channel $\theta_1$. The channels $\theta_1$ were intensively studied for the asymptotic quantum Birkhoff theorem (see \cite{haa}). We call $\theta_f$ VN-channels. One can understand that $f$, chosen from the von Neumann algebra $N$, is a quantum parameter of $\theta_f$. Note that in this setting the dimensions $H_{A'}=H_{B}=l_2^m$ coincide. Our first main theorem is the following comparison property on Schatten-$p$ norms for some nice classes of VN-channels.
\begin{theorem}[Comparison Theorem] \label{comp}Let $\theta_f$ be the channel defined by \eqref{channel3}. Assume that $N$ and $U$ satisfy
the following assumptions,
\begin{enumerate}
\item[i)] there exists a subalgebra $M\subset M_m$ as a standard inclusion ;
\item[ii)] the unitary $U$ admits a tensor representation $U=\sum_i x_i\ten y_i\in M'\ten N$;
\item[iii)] the operator $B=\sum_i |x_i\ran\ten \lan y_i^*|\in B(L_2(N), L_2(M))$ satisfies  $BB^*= \pl  id_{L_2(M)}$.
 \end{enumerate}
 Then for any bipartite state $\rho^{AA'}$ in $S_1(H_A\ten H_{A'})$ with some reference system $A$,
\begin{align} \|(id_A\ten \theta_1)(\rho)\|_p\kl   \|(id_A\ten \theta_f)(\rho)\|_p\kl \|f\|_p  \|(id_A\ten \theta_1)(\rho)\|_p \pl \label{pin} \end{align}
holds for $1\le p\le\infty$.
\end{theorem}
The assumptions i), ii),  iii) are extracted from several concrete classes of channels, including the group channels and quantum group channels mentioned in the introduction. They are discussed in detials in Section 8.
\subsection{Upper estimates via Theorem \ref{comp}}
Now we translate the $L_p$-estimates \eqref{pin} into capacity bounds.
We will prove several capacity bounds assuming the ``comparison property'' Theorem \ref{comp}. Let us start with an immediate consequence.
\begin{cor}\label{entropyin} Under the assumptions of Theorem \ref{comp}, denote $\pl\si^{AB}_f=id_A\ten\theta_f(\rho^{AA'})$ and respectively $\pl\si^{AB}_1=id_A\ten\theta_1(\rho^{AA'})$ as the outputs. Then the following inequalities hold:
\begin{enumerate}
\item[i)]$H(AB)_{\si_1}-\tau(f\ln f)\le H(AB)_{\si_f}\le H(AB)_{\si_1}$;
\item[ii)]$I_c(A\ran B)_{\si_f}\le I_c(A\ran B)_{\si_1}+\tau(f\ln f)$;
\item[iii)]$I(A:B)_{\si_f}\le I(A:B)_{\si_1}+\tau(f\ln f)$.
\end{enumerate}
In particular, if $H_A$ is one dimensional, {\rm i)} implies
\[H(B)_{\si_1}-\tau(f\ln f)\le H(B)_{\si_f}\le H(B)_{\si_1}\pl .\]
\end{cor}
\begin{proof}
Thanks to Theorem \ref{comp} we have
\begin{align*}
\|\si_1^{AB}\|_p\kl  \|\si_f^{AB}\|_p\kl \|f\|_p  \pl \|\si_1^{AB}\|_p\pl.
\end{align*}
Taking the derivatives at $p=1$, we deduce that
 \begin{align*}
 H(AB)_{\si_f}
 = \lim_{p\to 1^+} \frac{1-\|\si_f\|_p}{p-1}
 \kl \lim_{p\to 1^+} \frac{1-\|\si_1\|_p}{p-1} = H(AB)_{\si_1} \pl ,
 \end{align*}
and conversely
 \begin{align*}
 H(AB)_{\si_f}
 &= \lim_{p\to 1^+} \frac{1-\|\si_f\|_p}{p-1}\ge \lim_{p\to 1^+} \frac{1-\|f\|_{p}\|\si_1\|_p }{p-1} \\
 &= \lim_{p\to 1^+} \frac{(1-\|f\|_{p})\|\si_1\|_p + (1-\|\si_1\|_p)}{p-1} \ge  H(AB)_{\si_1}-\tau(f\ln f) \pl .
 \end{align*}
This yields i). For ii), applying i) for the outputs on $B$ and $AB$ we get
\begin{align*}
 I_c(A\ran B)_{\si_f}&= H(B)_{\si_f}-H(AB)_{\si_f} \nonumber \\
 &\le H(B)_{\si_1}-H(AB)_{\si_1}+\tau(f\ln f)=I_c(A\ran B)_{\si_1}+\tau(f\ln f) \pl .\end{align*}
Since $I(A: B)_{\si_f}=H(A)_{\si_f}+I_c(A\ran B)_{\si_f}$ and  $H(A)_{\si_f}=H(A)_{\si_1}$, we prove iii).
 \end{proof}

\begin{rem} {\rm It is easy to check that the function $g(p)=\|f\|_p$ is differentiable and satisfies ${g'(1)=\tau(f\ln f)}$ for finite dimensional $N$. The expression  $-\tau(f\ln f)$ may be considered as a von Neumann entropy for normalized traces in von Neumann algebras and closely related to the Fuglede determinant, see e.g. \cite{determinant,noteonentropy}. The normalization
$\tau(1)=1$ is used in order to prevent cumbersome constants for the symbol $f=1$. For the reader more familiar with the usual trace on matrices, we note that if $N\subset M_n$ and
the normalized trace $\tau=\frac{tr}{n}|_N$ is the restriction of the normalized trace $\frac{tr}{n}$ on $M_n$, then $\frac{1}{n}f$ is a density in $M_n$ and
\[\tau(f\ln f)= \ln n - H(\frac{1}{n}f)\pl .\]}
\end{rem}

\begin{cor}\label{cbounds}
Under the assumptions of Theorem \ref{comp}, we have
\begin{enumerate}
\item[i)]$Q^{(1)}(\theta_f)\le Q^{(1)}(\theta_1)+\tau(f\ln f)$, $Q(\theta_f)\le Q(\theta_1)+\tau(f\ln f)$;
\item[ii)]$C_{EA}(\theta_f)\le C_{EA}(\theta_1)+\tau(f\ln f)$.
\end{enumerate}
\end{cor}
\begin{proof} Taking the supremums on the second inequality of Corollary \ref{entropyin}, we obtain the inequality of $Q^{(1)}$. For $Q$, we observe that our assumptions are stable under taking tensor products. More precisely, we have   \[(\theta_f)^{\ten k}(\rho^{A^k{A'}^k})\lel id_{A^k}\ten \tau^k(U^{\ten k}(\rho\ten f^{\ten_k}) {U^*}^{\ten k})\pl, \]
and all assumptions of Theorem \ref{comp} are satisfied for $N^{\ten k}$ and $U^{\ten k}$. Then applying the inequality of $Q^{(1)}$ on $\theta_f^{\ten k}\equiv \theta_{f^{\ten k}}$
\begin{align*}Q(\theta_f)=&\lim_{k\to\infty}\frac1k Q^{(1)}(\theta_f^{\ten k})=\lim_{k\to\infty}\frac1kQ^{(1)}(\theta_{f^{\ten k}})\le\lim_{k\to\infty}\frac1k[Q^{(1)}(\theta_{1^{\ten k}})+\tau(f^{\ten k}\ln f^{\ten k})]\\=&\lim_{k\to\infty}\frac1k(Q^{(1)}(\theta_{1}^{\ten k})+k\tau(f\ln f))=\lim_{k\to\infty}\frac1kQ^{(1)}(\theta_{1}^{\ten k})+\tau(f\ln f)=Q(\theta_1)+\tau(f\ln f)\pl,
\end{align*}
which proves i). The assertion ii) follows immediately from the third inequality of Corollary \ref{entropyin}. \end{proof}

We can prove similar capacity bounds for the potential quantum capacity $Q^{(p)}$. For that, we need suitable $L_p$-approximations of the ``one-shot'' expression $Q^{(1)}$. For a quantum channel ${\Phi: S(H_{A'})\to S(H_{B})}$ and $p>1$, we can define the following two families of approximation quantities:
\[
 Q_p^{(1)}(\Phi) = \sup_{\rho \mbox{ \tiny pure}} \frac{\|(id_A\ten \Phi)(\rho^{AA'})\|_p}{\| \Phi(\rho^{A'})\|_p}\pl ,\ \ \ Q_{p,d}^{(1)}(\Phi) = \sup_{\rho^{AA'},\pl|A|\le d} \frac{\|(id_A\ten \Phi)(\rho^{AA'})\|_p}{\| \Phi(\rho^{A'})\|_p}\pl .
\]
 For a fixed $d$ both expressions are related to  $Q^{(1)}$ by differentiation at $p=1$.
\begin{lemma}\label{differ} For a quantum channel $\Phi$,
 \begin{enumerate}
 \item[i)] $\lim_{p\to1^+}\frac{1}{p-1}(Q_p^{(1)}(\Phi)-1)=Q^{(1)}(\Phi)$\pl ;
{\rm ii)} $\lim_{p\to1^+} \frac{1}{p-1}(Q_{p,d}^{(1)}(\Phi)-1)\le Q^{(1)}(\Phi)$\pl .

 \end{enumerate}
\end{lemma}
\begin{proof}
The proof of i) is straightforward by uniform convergence of $\displaystyle \frac{1-\|\rho\|_p}{p-1}$ to $H(\rho)$ on the state space. For ii) we purify $\rho_{AA'}$ on a system $AA'F$ with ${|H_F|=|H_A||H_{A'}|=d|H_{A'}|}$ and then apply i),
 \[\lim_{p\to1^+} \frac{1}{p-1}(Q_{p,d}^{(1)}(\Phi)-1)\le \lim_{p\to1^+} \frac{1}{p-1}(Q_{p}^{(1)}(\Phi\ten tr_{|A'|d})-1) =Q^{(1)}(\Phi\ten tr_{|A'|d} )= Q^{(1)}(\Phi) \pl .\]
\end{proof}
%{\begin{rem}\label{qp1hat} {\rm We will encounter yet another %$L_p$-approximation of $Q^{(1)}$ by vector-valued $p$-norms (see Section 4 %for a formal definition). Namely,
 %\begin{align*}\hat{Q}_p^{(1)}(\Phi) &= \max_{\rho \mbox{ \tiny pure}} %\|(\Phi\ten id_A)(\rho^{A'A})\|_{S_1(H_B;S_p(H_A))}\pl,\ \
 %\lim_{p\to1^+} \frac{1}{p-1}(\hat{Q}_{p}^{(1)}(\Phi)-1)\lel  Q^{(1)}(\Phi) %\pl.
 %\end{align*}
%The vector-valued $(1,p)$ norm is closely related to the sandwiched R{\'e}nyi %conditional entropy $H_p(A|B)$ introduced in \cite{renyi} and \cite{WWY},
 %\begin{align*} \frac{p}{p-1}\ln \|\si^{BA}\|_{S_1(H_B;S_p(H_A))}= -H_p(A|B) \xrightarrow{p\to 1} -H(A|B)=I_c(A\ran B)\pl.\end{align*}
%The norm expression $\hat{Q}_p^{(1)}(\Phi)$ gives a new approach for the convexity of $Q^{(1)}$ in terms of channels.
%\label{convex}}
%\end{rem}}
\begin{prop}\label{potential} Under the assumptions of Theorem \ref{comp}, we have
\[Q^{(p)}(\theta_f)\le \tau(f\ln f)+Q^{(p)}(\theta_1) \pl .\]
\end{prop}
\begin{proof} Let $\Psi:S_1(H_{A_1'})\to S_1(H_{B_1})$ be an arbitrary channel and $\rho^{AA'A_1'}$ be a purification of the bipartite state $\rho^{A'A_1'}$. Let us  denote by  $\omega^{AA'B_1}= id_{AA'}\otimes \Psi(\rho^{AA'A_1'})$ and \[{\si_f^{ABB_1}=id_{A}\otimes\theta_f\otimes \Psi(\rho^{AA'A_1'})}\ , \ \si_1^{ABB_1}=id_{A}\otimes\theta_1\otimes \Psi(\rho^{AA'A_1'}).\] Note that $\si_f^{ABB_1}=id_{AB_1}\otimes \theta_f (\omega_{AA'B_1})$, then we deduce, with the help of Theorem \ref{comp}, that
 \begin{align*}
 \|\si_f^{ABB_1}\|_p&\le \|f\|_p
 \|\si_1^{ABB_1}\|_p \kl \|f\|_p \pl Q^{(1)}_{p,d}(\theta_1\ten \Psi) \pl \|\si_1^{BB_1}\|_p
 \kl
  \|f\|_p \pl Q^{(1)}_{p,d}(\theta_1\ten \Psi) \pl
 \|\si_f^{BB_1}\|_p \pl .
 \end{align*}
Here $d=|A|$ and $Q^{(1)}_{p,d}$ appears because $\omega^{AA'B_1}= id_{AA'B_1}\otimes \Psi(\rho^{AA'A_1'})$ may not be a pure state. According to Lemma \ref{differ}, differentiating the inequality above  yields
 \begin{align*}
 Q^{(1)}(\theta_f\ten \Psi)
 &\kl  \tau(f\ln f)+ Q^{(1)}(\theta_1\ten \Psi)
 \kl \tau(f\ln f) + Q^{(p)}(\theta_1)+Q^{(1)}(\Psi) \pl .
 \end{align*}
 Since $\Psi$ is arbitrary, we deduce
 \begin{align*}
 Q^{(p)}(\theta_f)&=\sup_{\Psi}\ Q^{(1)}(\theta_f\ten \Psi)
 -Q^{(1)}(\Psi) \le
 \tau(f\ln f) + Q^{(p)}(\theta_1)\pl . \qedhere
 \end{align*}
 \end{proof}

%\begin{rem} {\rm  Later we will show in Section.3 under assumption of Corollary \ref{capbound}, $\theta_1=\E_{M}$ a conditional expectation onto $M=\oplus_k M_{n_k}\sum M_m$. Then it is not hard to show that
 %\[ \max_k \ln n_k=Q^{(1)}(\E_{M})\le Q(\E_{M})\le V^{(1)}(\E_{M})\lel \max_k \ln n_k \pl .\]
%Thus we can strengthen Theorem \ref{mainsquare}ii) to
 %\[ \max\{\max_k \ln n_k,\tau(f\ln f)\}\le Q^1(\theta_f)\le Q(\theta_f) \le V^1(\theta_f)\kl \tau(f\ln f)+\max_k \ln n_k \pl .\]
%we see that regularization are not so threatening any more in this particular context. }
%\end{rem}
We conclude this section by the application on the quantum dynamic capacity region. Although it is in general difficult to describe this capacity region exactly, there is a mathematically nice way to characterize the  ``one-shot, one-state'' region $C^{(1)}_{CQE,\si}$. Let us consider the cone
 \[ W \lel \{(C,Q,E)\pl| \pl 2Q+C\le 0,\ Q+E\le 0,\ Q+E+C\le 0\}\]
obtained from trading resources, i.e.
 teleportation, superdense coding and entanglement distribution (see \cite{HW} for a detailed explanation).
Given an output state
\[ \sigma^{XABE}=\sum_{x} p(x)|x\ran \lan x|^X\ten (1_A\ten V)\rho^{AA'}(1_A\ten V^*)\pl \]
where $V$ is the Stinespring partial isometry, we find the ``one-shot, one-state'' achievable region is
   \[ C^{(1)}_{CQE,\si}=
    (I(X;B)_{\si},\frac12I(A;B|X)_{\si},-\frac{1}{2}I(A:E|X)_{\si} )+ W \pl .\]
Thus, instead of estimating the entire ``one-shot'' region $C^{(1)}_{CQE}=\cup_{\si}C^{(1)}_{CQE,\si}$, we may compare the entropy terms $(I(X;B)_{\si},\frac12I(A;B|X)_{\si},-\frac{1}{2}I(A:E|X)_{\si} )$ for a single $\si$.

\begin{prop} Under the assumptions of Theorem \ref{comp}, denote $\tau=\tau(f\ln f)$, we have the following inclusions:
\begin{enumerate}
\item[i)]For each input $\rho^{XAA'}=\sum_x p(x)\ketbra{x}^X\ten \rho_x^{AA'}$ with $\rho_x^{AA'}$ pure states,\\ $C^{(1)}_{CQE,\si_f}(\theta_f)\subset C^{(1)}_{CQE,\si_1}(\theta_1)+(\tau,\frac{\tau}{2},\frac{\tau}{2})$;
\item[ii)]$C^{(1)}_{CQE}(\theta_f)\subset C^{(1)}_{CQE}(\theta_1)+(\tau,\frac{\tau}{2},\frac{\tau}{2})$, $C_{CQE}(\theta_f)\subset C_{CQE}(\theta_1)+(\tau,\frac{\tau}{2},\frac{\tau}{2})$.
\end{enumerate}
\end{prop}

\begin{proof}Let us first compare the rate triple $\big(I(X;B),\pl\frac12I(A;B|X),\pl-\frac{1}{2}I(A; E|X)\big)$ between $\si_f$ and $\si_1$. We denote them respectively as $(C_f, Q_f, E_f)$ and $(C_1,Q_1,E_1)$. By Corollary \ref{entropyin}, we have
\[ I(X;B)_{\sigma_f}\le \tau+ I(X;B)_{\sigma_1} \pl .
\]
Hence $C_f=C_1+\tau-\al_1$ for some $\al_1\gl 0 $. Similarly, for $\si^{AB}_{x,f}=id_A \ten \Phi (\rho_x^{AA'})$ we have
 \begin{align*} I(A;B|X)_{\si_f} \lel \sum_x p(x) H(\rho^A_x)
 + \sum_x p(x)[H(\si^{B}_{x,f})-H(\si^{AB}_{x,f})]
  \end{align*}
and
 \begin{align*} I(A;E|X)_{\si_f} \lel \sum_x p(x) H(\rho^A_x)
 + \sum_x p(x)[H(\si^{E}_{x,f})-H(\si^{AE}_{x,f})]
   \pl .\end{align*}
Since each $\rho^{AA'}_x$ is pure we get
$ H(\si^{E}_{x,f})=H(\si^{AB}_{x,f})$ and $
 H(\si^{AE}_{x,f})=H(\si^{B}_{x,f})$.
 This means
 \begin{align*} Q_f \lel Q_1+ \frac{\tau-\al_2}{2} \pl ,\pl E_f \lel E_1+\frac{\tau-\al_2}{2} \pl \end{align*}
for some $\al_2\gl 0$. Now we observe that $(-\al_1,-\frac{\al_2}{2},-\frac{\al_2}{2})\in W$ because $-\al_1-\al_2 \kl 0$ and $   -\al_2\le 0$. Thus we obtain
 \[ (C_f,Q_f,E_f) \in (\tau,\frac{\tau}{2},\frac{\tau}{2})+ (C_1,Q_1,E_1)+W \pl .\]
Since $W$ is a cone, $W+W=W$. we get
\[C^{(1)}_{CQE,\si_f}=(C_f,Q_f,E_f)+W\subset (\tau,\frac{\tau}{2},\frac{\tau}{2})+ (C_1,Q_1,E_1)+W+W= (\tau,\frac{\tau}{2},\frac{\tau}{2})+C^{(1)}_{CQE,\si_1}\pl .\]
This concludes the proof of i). For ii), taking the union over all output $\si$ implies
 \begin{align*} C^{(1)}_{CQE}(\theta_f)\subset (\tau,\frac{\tau}{2},\frac{\tau}{2}) + C^{(1)}_{CQE}(\theta_1)\pl .\end{align*}
For iii), we use again the fact that $\theta_{f}^{\ten_k}\equiv\theta_{f^{\ten_k}}$ is of the same nature as $\theta_f$ and hence we deduce that
 \begin{align*}\frac{1}{k} C^{(1)}_{CQE}(\theta_f^{\otimes k})\subset \frac{1}{k}[k(\tau,\frac{\tau}{2},\frac{\tau}{2})+C_{CQE}^{(1)}(\theta_1^{\otimes k})]=(\tau,\frac{\tau}{2},\frac{\tau}{2})+\frac{1}{k}C_{CQE}^{(1)}(\theta_1^{\otimes k})\pl.\end{align*}
The result follows by taking the union over $k\in \nz$. \end{proof}
\begin{rem}{\rm i) All above estimates rely on the special channel $\theta_1$. Fortunately, we will see in Section 5 that $\theta_1$ is a channels as direct sums of partial trace, which has clear capacity expression depending on the von Neumann algebra $M$. It can also be deduced from \cite{HW} that the capacity region of such $\theta_1$ is strongly additive, hence it is regularized. Namely, we obtain the following ``single-letter upper bound''
\begin{align*}C_{CQE}(\theta_f)\subset (\tau,\frac{\tau}{2},\frac{\tau}{2})+
 {C_{CQE}^{(1)}(\theta_1)} \pl .\end{align*}
ii) If in additional $\theta_f$ is unital ($\theta_f(1)=1$), we find ${(\tau,\frac{\tau}{2},\frac{\tau}{2})\in 2C^{(1)}_{CQE}(\theta_f)}$. Indeed, we choose the input state $\rho^{AA'}$ to be a maximal entangled state, then  ${\displaystyle (0,\frac12(\ln m+\tau),\frac12(-\ln m+\tau))}$ and hence $(\tau,0,0)$, $(0,\frac{ \tau}{2},\frac{\tau}{2})$ belong to $C^{(1)}_{CQE}(\theta_f)$. Our estimate implies a comparison of convex regions often considered in convex geometry and Banach spaces
  \[  C_{CQE}(\theta_1)\subset C_{CQE}(\theta_f)\subset (\tau,\frac{\tau}{2},\frac{\tau}{2})+
 C_{CQE}(\theta_1)\subset 3 C_{CQE}(\theta_1) \pl .\]
 The first inclusion is an immediate consequence of Lemma \ref{mu}.
}\end{rem} 

%% file: Qcap_operatorspace.tex
\section{Operator space duality and $L_p$-spaces}
\subsection{Basic operator space} The background on operator space reviewed here is avalible in \cite{ER} and \cite{Psbook}. We say X is a (concrete) \emph{operator space} if $X\subset B(H)$ is a closed subspace for some Hilbert space $H$. The $C^*$-algebra $B(H)$ has a natural sequence of matrix norms associated with it: $M_n(B(H))=B(H^{\otimes n})$. Then the inclusion $X\subset B(H)$ not only equips $X$ with a Banach space norm, but also a sequence of norms on the vector-valued matrices
\[M_n(X)=\{(x_{ij})_{ij}|\pl x_{ij}\in X, \forall 1\le i,j \le n\}\pl. \]
Here we understand $M_n(X)\subset M_n(B(H))$ as being isometrically embedded. This sequence of matrix norms satisfy Ruan's Axioms, which are two properties inherited from $M_n(B(H))$ (here $1$ denotes the identity operator of $B(H)$):
\begin{enumerate}
\item[i)]$\text{For any}\  a,b\in M_n,\ x=(x_{ij})\in M_n(X),\ $
\begin{center}$\norm{(a\otimes 1)(x_{ij})(b\otimes 1)}{}\le \norm{a}{M_n}\norm{x}{M_n(X)}\norm{b}{M_n}$;
\end{center}
\item[ii)]$\text{For any}\  x=(x_{ij})\in M_n(X),\  y=(y_{ij})\in M_m(X), \ $
\begin{center}$\norm{\left( \begin{array}{cc}
x& 0  \\
0 & y \end{array} \right)}{M_{n+m}(X)}\le \max\{\norm{x}{M_n(X)},\ \norm{y}{M_m(X)}\} \pl .$ \end{center} \label{RA}
\end{enumerate}
 An operator space structure is either given by a concrete embedding $X\subset B(H)$ or a sequence of matrix norms satisfying Ruan's axioms. Thanks to Ruan's theorem this defines the same category, i.e. every matrix normed space satisfying Ruan's axioms admits an embedding ${\iota:X\to B(H)}$ which preserves the norms on all levels. A map ${\iota:X\to Y}$ such that ${id_n\ten \iota:M_n(X)\to M_n(Y)}$ is isometric for all $n$ is called a \emph{complete isometry}. Basic examples of operator spaces are given by the column space $C_n$ and  the row space $R_n$:
\begin{align} C_n={\rm span}\{e_{i,1}|1\le i \le n\}\subset M_n,\ \  R_n={\rm span}\{e_{1,i}|1\le i \le n\}\subset M_n\pl .\end{align}
Here and in the following $e_{i,j}$ denote the standard matrix unit (with the respect to the computational basis), i.e. the matrix which is $0$ except for the single entry $1$ in $i$-th row and $j$-th column. A basis-free description of the row and column space can be given as follows
\begin{align}H^c=B(\cz,H)\ \ ,\ H^r=B(H,\cz)\pl .\end{align}
The morphisms between operator spaces are completely bounded maps ($cb$-maps). Given two operator spaces $X,Y$ and a linear map $u:X\to Y$, we say $u$ is \emph{completely bounded} if the $cb$-norm
 \begin{align}\label{cbnorm}
  \|u\|_{cb} \lel \sup_n \|id_{M_n}\ten u:M_n(X)\to M_n(Y)\|
 \end{align}
is finite. The space of completely bounded maps from $X$ to $Y$ is denoted as $CB(X,Y)$. Clearly, $CB(X,Y)$
 is a Banach space, even more  an operator space equipped with the matrix level structure $M_n(CB(X,Y))=CB(X,M_n(Y))$. Particularly, ${X^*= CB(X,\cz)}$ is called the operator space dual of $X$. %Throughout this paper we use the anti-linear bracket
%\[ \lan a,b\ran \lel tr(a^*b) \]
%between the space of bounded operators $B(H)$ and its predual the space of trace class operators $S_1(H)$. Formally this leads to the dual relation $\overline{S_1(H)}^*=B(H)$. In operator space theory, it is customary to use the linear duality bracket
 %\[ (a,b) \lel  tr(a^{t}b) \pl. \]
%Here $a^{t}$ is the transposed of $a$. The advantage of this bracket is the bilinearity of dual action. However, the anti-linear bracket is of no harm because $\overline{S_1(H)}=S_1(H)$ as sets, and the complex conjugate map is an anti-linear complete isometry. We will extend this observation to all duality relations considered in this text.
\subsection{Haagerup tensor product}Beyond the basic operator space concepts, the Haagerup tensor product is also a key tool in our estimates. Let us recall that for two operator spaces $X\subset B(H)$ and $Y\subset B(K)$, the \emph{Haagerup tensor} norm is defined on $X\otimes Y$ as
\[ \|z\|_{X \ten_h Y} \lel \inf_{z=\sum_k x_k\ten y_k}
 \|(\sum_k x_kx_k^*)^{1/2}\|_{B(H)} \|(\sum_k y_k^*y_k)^{1/2}\|_{B(K)} \pl .\]
In many cases we will not be able to provide a concrete embedding $X\subset B(H)$, and then it is better to note that
 \begin{align*} \|(\sum_k x_kx_k^*)^{1/2}\| \lel \|\sum_{k} x_k \ten e_{1,k} \|_{R_n(X)} \pl ,\pl \|(\sum_k y_k^*y_k)^{1/2}\|
 \lel \|\sum_k e_{1,k}\ten y_k\|_{C_n(X)} \pl ,\end{align*}
where $C_n(X), R_n(X)\subset M_n(X)$ are the $X$-valued column and row spaces. The Haagerup tensor product can recover the operator space structure
\begin{align*}
M_n(X) \lel C_n\ten_h X\ten_h R_n\pl,\pl\ C_n(X)=C_n\ten_h X \pl,\pl
R_n(X)=X\ten_h R_n \pl ,
\end{align*}
which holds completely isometrically. In particular, we have
\begin{align*}
M_n(M_m) &= C_n\ten_h M_m\ten_h R_n\lel M_{mn} \\
 C_n(C_m)&=C_n\ten_h C_m= C_{mn} \pl,\pl
R_n(R_m)=R_m\ten_h R_n = R_{mn} \pl .
\end{align*}
These identifications are also compatible with the general duality  relation
 \begin{align*} (X\ten_h Y)^* \lel X^*\ten_h Y^* \pl .\end{align*}
We recall that (see e.g. \cite{ER, Psbook}) $C_n^*=R_n$, $R_n^*=C_n$  holds completely isometrically. This implies
 \begin{align*}
 M_n^{*} =  (C_n\ten_h R_n)^* = R_n\ten_h C_n = S_1^n\ , \
 (S_1^n)^{*}  = (R_n\ten_h C_n)^* = C_n \ten_h R_n = M_n \pl .
 \end{align*}
It is important to note that the columns in $S_1^n$ carry the operator space structure of $R_n$, and the rows in $S_1^n$ become $C_n$. %The space of trace class operators $S_1(H)$ can be viewed as generalized states, or more precisely the set of densities for linear combination of not necessarily pure states.
Another fundamental concept is the \emph{minimal tensor} norm for operator spaces $X\subset B(H)$, $Y\subset B(K)$ given by
 \[ X\ten_{\min}Y\subset B(H)\ten_{\min}B(K)\subset
  B(H\ten K) \pl ,\]
where the second inclusion serves as a definition of the $\min$-norm ($\min$ operator space structure). The connection with the space $CB(X,Y)$ is functorial, i.e. if one of the spaces is finite dimensional then
 \begin{align}\label{min} CB(X,Y) \lel X^*\ten_{\min}Y  \end{align}
holds completely isometrically. The minimal tensor norm is the smallest operator space tensor norm (see \cite{Psbook,ER}). %Hence we have the following fact from \cite{blecher}.
%\begin{fact}\label{hmin} The Haagerup tensor product is an operator space tensor product, in particular the inclusion
% \[ X\ten_hY \subset X \ten_{\min}Y \]
%is completely contractive.
%\end{fact}

\subsection{Complex interpolation}
Let $X_0$ and $X_1$ be two Banach spaces. We say $X_0$ and $X_1$ are compatible if there exists a Hausdorff topological vector $X$ such that $X_0, X_1\subset X$ as subspaces. One can define the sum as
\[X_0+X_1:\lel \{x\in X|x=x_0+x_1\pl \text {for some}\pl x=X_0, x_1\in X_1\}\pl,\]
and $X_0+X_1$ equipped with the norm
\[\norm{x}{X_0+X_1}=\inf_{x=x_0+x_1} (\norm{x_0}{X_0}+\norm{x_1}{X_1})\]
is again a Banach space. Let us denote by $S=\{z|0\le Re (z)\le 1\}$ the classical vertical strip of unit width on the complex plane and $S_0=\{z|0< Re (z)< 1\}$ its open interior. We will consider the space $\F(X_0, X_1)$ of all functions $f:S\to X_0+X_1$, which are bounded and continuous on $S$ and analytic on $S_0$, and moreover
\[\{f(it)|t\in \mathbb{R}\}\subset X_0\pl ,\pl \{f(1+it)|t\in \mathbb{R}\}\subset X_1\pl.\]
$\F(X_0, X_1)$ is a Banach space under the norm
\[\norm{f}{\F}=\max\{\sup_{t\in \mathbb{R}} \norm{f(it)}{X_0}, \sup_{t\in \mathbb{R}}\norm{f(1+it)}{X_1}\}\pl. \]
For $0<\theta<1$, the complex interpolation space $(X_0,X_1)_\theta$ is defined as a subspace of $\F(X_0,X_1)$ as follows
\[(X_0, X_1)_\theta=\{x\in X_0+X_1| \pl x=f(\theta), f\in F(X_0, X_1)\} \pl.\]
$(X_0,X_1)_\theta$ is a Banach space equipped with the norm
\[\norm{x}{\theta}=\inf \{\norm{f}{\F}| f(\theta)=x\}\pl .\]
For example, the Schatten-$p$ class is the interpolation space of bound operator and trace class
\[S_p(H)=(B(H), S_1(H))_{\frac{1}{p}}\pl.\]
The following Stein's interpolation theorem (cf. \cite{BL}) is a key tool in our analysis. \begin{theorem}\label{stein}
Let $(X_0,X_1)$ and $(Y_0,Y_1)$ be two compatible couples of Banach spaces. Let $\{T_z| z\in S\}\subset B(X_0+X_1, Y_0+Y_1)$ be a bounded analytic family of maps such that
\[\{T_{it}|t\in \mathbb{R}\}\subset B(X_0,Y_0)\pl ,\pl \{T_{1+it}|t\in \mathbb{R}\}\subset  B(X_1,Y_1)\pl.\]
Suppose $M_0=\sup_t{\norm{T_{it}}{B(X_0,Y_0)}}$ and  $M_1=\sup_t{\norm{T_{1+it}}{B(X_1,Y_1)}}$ are both finite, then $T_\theta$ is a bounded linear map from $(X_0,X_1)_\theta$ to $(Y_0,Y_1)_\theta$ and
\[\norm{T_\theta}{B((X_0,X_1)_\theta ,(Y_0,Y_1)_\theta)}\le M_0^{1-\theta}M_1^{\theta}\pl .\]
\end{theorem}
\noindent In particular, when $T$ is a constant map, the above theorem implies
\begin{align}
\norm{T}{B((X_0,X_1)_\theta ,(Y_0,Y_1)_\theta)} \le \norm{T}{B(X_0 ,Y_0)}^{1-\theta}\norm{T}{B(X_1 ,Y_1)}^{\theta} \pl . \label{interpolation}
\end{align}
\subsection{Noncommutative $L_p$-spaces}
Noncommutative $L_p$-spaces may be obtained by complex interpolation. Indeed, (for finite dimension $H$) we have
 \begin{align*} S_p(H) \lel (B(H),S_1(H))_{\frac{1}{p}}
 \lel (H^c\ten_h H^r, H^r\ten_h H^c)_{\frac{1}{p}}
 \lel (H^c, H^r)_{\frac1p}\ten_h (H^r, H^c)_{\frac{1}{p}} \pl.  \end{align*}
 The second equality is an instance of Kouba's interpolation formula for the Haagerup tensor product (see \cite{BL,pvp,Psbook} for more details),
 \[(X_0, X_1)_{\theta}\ten_h (Y_0, Y_1)_{\theta}=(X_0\ten_h Y_0, X_1\ten_h Y_1)_{\theta}\pl.\]
 We will adapt the notation $H^{c_p}=(H^c, H^r)_{\frac{1}{p}}$ and $H^{r_p}=(H^r, H^c)_{\frac{1}{p}}$ for the columns and row in $S_p(H)$ respectively. This definition leads to %$(H^{c_p})^*=H^{r_p}=H^{c_{p'}}$ for $1/p+1/p'=1$ and
 the ``little Fubini theorem''
 \begin{equation} \label{hcp}
 H^{c_p}\ten_h K^{c_p}\lel (H\ten K)^{c_p}\pl,\pl
 H^{r_p}\ten_h K^{r_p} \lel (H\ten K)^{r_p} \pl ,
 \end{equation}
 for two Hilbert spaces $H$ and $K$.
%Given a positive operator $\rho\in B(H)$, if $\rho=\eta\eta^*$ with $\eta \in B(H)$, then
%\[\norm{\rho}{S_p(H)}=\norm{\eta}{S_{2p}(H)}^2.\]
%This could be rephrased by Haagerup tensor product,
%\[\norm{\rho}{H^{c_p}\ten_h H^{r_p}}=\norm{\eta}{H^{c_p}\ten_h H^{r_p}}^2=\norm{\eta}{S_{2p}(H)}^2 \pl.\]
In some instance we will make use of vector-valued $L_p$ spaces. For an operator space $X$, we recall Pisier's definition
 \begin{align*} S_p(H,X) \lel H^{c_p} \ten_h X\ten_h H^{r_p}\pl .\end{align*}
An important special case is given by
 \begin{equation}\label{QPP}
  \|\xi\|_{S_p(H_A,S_q(H_B))} \lel
  \sup_{\|a\|_{2r}\|b\|_{2r}\le 1} \|(a\ten 1_B)\xi(b\ten 1_B)\|_{S_q(H_A\ten H_B)}
  \end{equation}
where $q\le p$, $1/p+1/r=1/q$ and
 \begin{equation*}\label{PQQ}
  \|\xi\|_{S_p(H_A,S_q(H_B))} \lel
  \inf_{\xi=(a\ten 1_B)\eta(b\ten 1_B)} \|a\|_{2r}\|\eta\|_{S_q(H_A\ten H_B)}\|b\|_{2r}
   \end{equation*}
where $q\ge p$, $1/q+1/r=1/p$. It is not difficult to show that for $\xi\gl 0$ it suffices to consider $a=b^*\ge 0\in B(H_A)$ in both cases.
% Thanks to \eqref{hcp} and the associativity of the Haagerup tensor product we have a Fubini-type theorem
% \begin{align}\label{shuffle} S_p^n(S_p^m)\lel C_p^n\ten_h
% C_p^m\ten_hR_p^n\ten R_p^m
% \lel C_p^{nm}\ten_h R_{p}^{nm}\lel S_p^{nm} \pl .\end{align}
% This corresponds to the case $p=q$ in \eqref{QPP} and \eqref{PQQ}.\\

%% file: Qcap_stinespringspaces.tex
\section{Stinespring space and its Operator Space structures}
\noindent Suppose a channel $\Phi: S_1(H_{A'})\to S_1(H_{B})$ from Alice to Bob has a Stinespring dilation \[\Phi(\rho)=id_B\otimes tr_E (V \rho V^*)\pl ,\]
where $V: H_{A'}\to H_{B}\otimes H_E$ is a partial isometry such that $V^*V=1_{A'}$. Then the \emph{Stinespring space} of $\Phi$ is defined to be the range of partial isometry $V$:
\[ \st(\Phi)  \equiv \text{Im} (V)= \{V(h)  | \pl h\in H_A \} \subset H_B\ten H_E \pl.\]
Although the partial isometry $V$ is not unique, different dilations only differ by unitary transformations on $H_E$, and hence will not affect the operator space structure of $\st(\Phi)$. The Stinespring space is well-known and has been used instrumentally in disproving the additivity conjecture for the minimal entropy (see \cite{hayden}). It has become clear that the family of Schatten $p$-norms on $H_B\ten H_E$ are related to entropy. In this paper we will go one step further and consider the operator space structure of the Stinespring space. For $ 1\le p\le \infty$, let us denote $\st_p(\Phi)$ as the operator subspace $\st(\Phi)$ induced by the following inclusion
\[\st_p(\Phi) \subset H_B^{c_p}\ten_{h} H_E^{r}  \pl .\]
Let us recall that for two Hilbert space $H$ and $K$,
 \[ H^{c_p}\ten_h K^r \lel  [H^c\ten_h K^r, H^r\ten_h K^r]_{\frac1p} \lel S_{2p}(K, H) \pl .\]
Here $S_{p}(H, K)$ stands for Schatten-$p$ class of operators from $K$ to $H$. Note that the operator space structure here is not usual one (i.e. $H^{c_{2p}}\ten_h K^{r_{2p}}$), see \cite{JP} for more details on asymmetric $L_p$-spaces.

\begin{lemma}\label{shuffle}  Let $\Phi:S_1(H_{A'}) \to S_1(H_B))$ be a channel with Stinespring dilation isometry $V$. Let $\xi^{AA'}$ and $\rho^{AA'}=\xi\xi^*$ be operators in $B(H_{A}\ten H_{A'})$. Denote ${\eta=(1_A\ten V)\xi }$, then
 \begin{enumerate}
\item[i)] $\|(id_A\ten\Phi)(\rho^{AA'})\|_{S_p(H_A\ten H_B)} \lel \|\eta\|^2_{H_A^{c_p}\ten_h \st_p(\Phi)\ten_h (H_{A'}\ten H_A)^{r}}$;
 \item[ii)] $\|\Phi(\rho^{A'})\|_{S_p(H_B)} \lel \|\eta\|^2_{\st_p(\Phi)\ten_h (H_A\ten H_{A'}\ten H_A)^{r}}$.
       \end{enumerate}
In particular, if $\rho= |\xi\ran \lan \xi|$ is given by a pure state then $\eta$ belongs to $H_A^{c_p}\ten_h \st_p(\Phi)$ for i) and respectively  $\st_p(\Phi)\ten_h H_A^{r}$ for ii).
\end{lemma}

\begin{proof} In this proof, it is important to track the position of vectors and covectors (column vectors and row vectors) in the tensor components. We may assume that $\Phi$ has Kraus operators $\Phi(\rho)=\sum_i x_i\rho x_i^*$, and $V=\sum_i x_i\ten e_{i,1}$. To specify the tensor components, we denote ${\xi=\sum_{j} |a_j^A\ran |h_j^{A'}\rangle \langle b_j^{A}|\lan k_j^{A'}|}$ where $|a_j^A\ran, |b_j^A\ran$ are vectors of $ H_{A}$ and $|h_j^{A'}\rangle, |b_j^{A'}\rangle$ vectors of $H_{A'}$. We use the ``little Fubini theorem'' \eqref{hcp}
 \begin{align*}
 & \eta \lel \sum_{i=1}^d (1_A\ten x_i)\xi^{AA'}\otimes e_{i,1}
  =\sum_{j,i} |a_j^A\rangle|x_i(h_j)^B\rangle  \ten \langle b_j^A,k_j^{A'}|\ten\ket{i^E}\\
 &\cong \hat{\eta}\equiv\sum_j |a_j^A\rangle \ten \bigg(\sum_i |x_i(h_j)^B\ran \ten \langle i^E|\bigg) \ten \langle k_j^A,b_j^{A'}|\ \ \ \ \ \ \ \text{ (shuffle)}\\
 & \in H_A^{c_p}\ten_h \st_p(\Phi)\ten_h (H_{A} \ten H_{A'})^{r}
 \subset(H_A\ten H_B)^{c_p}\ten_h (H_E\ten H_{A}\ten H_{A'})^{r} \pl ,
 \end{align*}
where in the second line above, we first change the role of $E$ system from column to row, and then switch between row vectors $\bra{i^E}$ and $\langle k_j^A,b_j^{A'}|$. This action is an identification and we get $\hat{\eta}\hat{\eta}^*=(id_A\ten \Phi)(\rho_{AA'})$. Now the first assertion follows from the fact ${\norm{a}{S_{2p}(K, H)}^2=\norm{aa^*}{S_p(H)}}$. For ii), we first note that
  \begin{align*} \|\Phi(\rho_{A'})\|_p
 \lel \|(tr_A\ten id_B)\circ\Phi(\rho^{AA'})\|_p
 \lel \|tr_A\ten id_B(\hat{\eta}\hat{\eta}^*)\|_p \pl.  \end{align*}
 The trace on $A$ make $H_A$ row vector to the right of $\st_p(\Phi)$. Namely,
 \begin{align*}
 & \eta \cong \tilde{\eta}\equiv\sum_j  \bigg(\sum_i |x_i(h_j)^B\ran \ten \langle i^E|\bigg) \ten \lan a_j^A|\ten \langle k_j^A,b_j^{A'}|\ \ \ \ \ \ \ \text{ (shuffle)}\\
 & \in  \st_p(\Phi)\ten_h (H_A \ten H_{A} \ten H_{A'})^{r}
 \subset H_B^{c_p}\ten_h (H_E\ten H_A\ten  H_{A}\ten H_{A'})^{r} \pl ,
 \end{align*}
When $\rho\in S_1(H_{A'})$ is a pure state, the right part $(H_{A}\ten H_{A'})^{r}$ become trivial, which yields the last assertion.\qd

Let us recall another definition from the theory of noncommutative vector-valued $L_p$ space. For an operator space $X$ we use
 \[ C_p^n(X) \lel C_p^n\ten_h X \quad, \quad R_p^n(X) \lel X\ten_{h} R_p^n \pl .\]
In particular, $R_n(X)=X\ten_h R_n$ are the rows for $X$. The space $\pl C_p^n(X)$ may be understood as the columns in the the vector-valued space $S_p^n(X)=C_p^n\ten_h X\ten_h R_p^n$. We define the row-column $p$-concavity for $X$ by
 \[ \crp_p(X) \lel \sup_n \|id_n\ten id_{X}:R_n(X)\to C_p^n(X) \| \pl .\]
 The next proposition provides the link between operator spaces structures and the ``one-shot'' expression $Q^{(1)}$.
\begin{prop}\label{rcq1}
For a channel $\Phi$, $Q_p^{(1)}(\Phi)=rc_p(st_p(\Phi))^2 $. \label{RC}
\end{prop}

\begin{proof} Use the definition, we have
\[Q_p^{(1)}(\Phi) = \sup_{\rho \mbox{ \tiny pure}} \frac{\|(id_A\ten \Phi)(\rho^{AA'})\|_p}{\| \Phi(\rho^{A'})\|_p}\pl = \sup_{\eta } \frac{\|\hat{\eta}\|_{H_A^{c_p}\ten_h\st_p(\Phi)}^2}{\|\tilde{\eta}\|_{\st_p(\Phi)\ten_h H_A^r}^2}= rc_p(st_p(\Phi))^2 \pl ,
\]
where the supremum runs over $\eta\in H_A\ten st(\Phi)$.
According to Lemma \ref{shuffle}, we know that a pure state $\rho$ corresponds to an element $\eta \in H_A\ten\st_p(\Phi)$.
\end{proof}
\begin{rem}{\rm  For a subspace $X\subset H^{c_p}\ten_h K^{r}$, it is easy to see that $\crp_p^2(X)$ is the smallest constant $C$ such that
\[\|\sum_k x_k^*x_k\|_p \kl C \pl \|\sum_k x_kx_k^*\|_p \]
holds for all finite sequences $(x_k)\in X$. Clearly, this is a measure of non-commutativity. }
\end{rem}
For the rest of this section, let us fix the notation $1\le p \le \infty ,\pl 1/p+1/p'=1$. We illustrate the row-column $p$-concavity on some  elementary examples.
\begin{exam}{\rm \label{ex1}Let $M_{m,d}=C_m\ten_h R_d$ be the $m\times d$ matrix space and $S_{2p}^{m,d}=C^m_p\ten_h R_d$. Then $\crp_{p}(S_{2p}^{m,d}) \lel m^{1/2p'}$. This implies that for the partial trace map $id_m\ten tr_d: M_m\ten M_d \to M_m$, $Q_p^{(1)}(id_m\ten tr_d)=m^{1/p'}$.}
\end{exam}
\begin{proof} We know the case $p=1$ is trivial, $\crp_{1}(X)=1$ for any operator space $X$. For $p=\infty$, we may consider
  \[ \xi \lel \sum_{1\le j\le n, \pl 1\le l \le m} e_{l,1}\ten   \xi_{l,j} \ten e_{1,j} \in C_{m}\ten_{h} R_{d}\ten_h R_n =M_{m,d}\otimes_h R_n\pl. \]
Then since $M_{m,d}\otimes_h R_n=M_{m,dn}$, we deduce that
  \begin{align*}\sup_{1\le l\le m} (\sum_{1\le j\le n} \|\xi_{l,j}\|_2^2)^{1/2}\kl \norm{\sum_{1\le l,l'\le m}(\sum_{1\le j\le n} (\xi_{l,j}^*\xi_{l,j'}))e_{l,l'}}{M_{m}}^{\frac{1}{2}} =\|\xi\|_{C_{m}\ten_h R_{dn}} \pl . \end{align*}
This implies
 \begin{align*}
 \|\sum_{j,\pl l} e_{j,1}\ten e_{l,1}\ten \xi_{l,j}\|^2_{C_n\ten_h C_m\ten _h R_d}
 = \|\sum_{j,\pl  l} \xi_{l,j}^*\xi_{l,j}\|
 \kl m\sup_{1\le l\le m}  \sum_{j} \|\xi_{l,j}\|^2
 \le m  \|\xi \|_{C_m\ten_h R_{dn}}^2 \pl .
 \end{align*}
Equality is obtained by looking at $n=m$,\
 $\xi=\sum_{l} e_{l,1}\ten e_{1,1}\ten e_{1,l}\in C_{m}\ten_h R_{d}\ten R_{n}$ which has norm $1$ and
 \begin{align}\label{nnorm}
  \|\sum_l e_{l,1}\ten e_{l,1}\ten e_{1,1}\|_{C_m\ten_h C_{m}\ten_h R_{d}}\lel \sqrt{m} \pl . \end{align}
Thus we have shown that $\crp_{\infty}(M_{m,d})=\sqrt{m}$. Since the subspace $C_p^m\ten_h R_d$ is complemented in $C_p^n\ten_h R_n$ ($m,d\le n$) with the same projection for all $1\le p\le \infty$, we apply interpolation \eqref{interpolation} and deduce $\crp_p(C_p^m\ten_h R_d)\le m^{1/2p'}$. The equality is obtained by same element as in \eqref{nnorm}. The last assertion follows from that $st_p(id_m\ten tr_d)=C_p^m\ten_h R_d$.
\end{proof}

\begin{exam}{\rm \label{ex2} Let $X_i\subset H_i^{c_p}\ten_h K_i^{r},\  1\le i\le m$ be a sequence of subspaces. Then the space
  \begin{align} \ell_{2p}\{X_i\} \lel \{\sum_{i=1}^m e_{i,1}\ten x_i\ten e_{1,i}| x_i\in X_i\} \subset (\ell_2\{H_i\})^{c_p}\ten_h (\ell_2\{K_i\})^{r} \end{align}
satisfies $\crp_p(\ell_{2p}^m\{X_i\})=\sup_ {i\le m}\crp_p(X_i)$. Moreover, given a finite sequence of quantum channel $\Phi_i: S_1(H_i)\to S_1(K_i)$, the direct sum channel $\oplus_i\Phi_i: S_1(\oplus_i H_i)\to S_1(\oplus_i K_i)$ satisfies $Q_p^{(1)}(\oplus_i\Phi_i)=\max_i Q_p^{(1)}(\Phi_i)$. By taking derivatives, we reproves the observation $${Q^{(1)}(\oplus_i\Phi_i)=\max_i Q^{(1)}(\Phi_i)}$$ in \cite{Wolf} via a different approach.}
\end{exam}
\begin{proof} Here we regard   $\ell_{2p}^m\{X_i\} \subset (\ell_2\{H_i\})^{c_p}\ten_h (\ell_2^m\{K_i\})^{r}$ as a block diagonal subspace. Thus $\crp_p(\ell_{2p}^m\{X_i\})\ge\sup_{i\le m}\crp_p(X_i)$ trivially holds. For the inverse inequality, let us first observe that
  \begin{align*} \|\sum_i e_{i,1}\ten x_i\ten e_{1,i}\|_{2p} \lel (\sum_j \|x_i\|_{2p}^{2p})^{1/2p} \pl . \end{align*}
This is obvious for $p=1$ and $p=\infty$ and then follows by interpolation (see also \cite{pvp,JP} for very similar/more general arguments). Now let $x=\sum_{i,l} e_{l,1}\ten e_{i,1}\ten x_{i,l} \ten e_{1,i}$, we find that
 \begin{align*}
 \| x\|_{R_n(\ell_{2p}\{X_i\})}
 &= \|\sum_{i\le m,l\le n} e_{i,1}\ten x_{i,l} \ten e_{1,i}\ten e_{1,l}\|
 \lel (\sum_{i=1}^{m} \| \sum_{l\le n} x_{i,l}\ten e_{1,l}\|^{2p}_{R_n(X)})^{1/2p} \\
 &\kl (\sum_{i=1}^{m} \crp_p(X_i) \| \sum_{l\le n}e_{l,1}\ten x_{i,l} \|^{2p}_{C_p^n(X)})^{1/2p} \\
 &\kl \sup_{1\le i\le m}\pl \crp_p(X_i) \|x\|_{C_p^n(\ell_{2p}^m\{X_i\})} \pl .
 \end{align*}
Here we used \eqref{hcp} $(\ell_2^m)^{c_p}\ten_h H^{c_p}\lel H^{c_p}\ten_h (\ell_2^m)^{c_p}$ and ${K^{r}\ten_h (\ell_2^m)^{r}= (\ell_2^m)^{r}\ten_h K^{r}}$. The last assertion follows from that the Stinespring space of direct sum channel is the direct sum of each Stinespring space. \end{proof}
\begin{exam} {\rm \label{ex3} Let $\Phi$ be channel and $\nen$. Then
 \[ \crp_p(\st_p(id_n\ten \Phi)) \lel n^{1/2p'}\st_p(\Phi) \pl .\]
 In particular, $Q_p^{(1)}(id_n\ten \Phi)=n^{1/2p'} Q_p^{(1)}(\Phi)$.}
\end{exam}

\begin{proof} Let $\Phi(\rho)=\sum_{k=1}^m x_k\rho x_k^*$ be a channel from $S_1(H_{A'})$ to $S_1(H_B)$. Then we see that
 \[ \st_p(id_n\ten \Phi)\lel \{\sum_{j=1}^n  e_j \ten x_k(h_j)\ten e_k | h_j\in H_{A'}\} \lel C_p^n\ten_h \st_p(\Phi) \pl .\]
Let us define $X=\st_p(\Phi)\ten_h R$, $Y=C_p\ten_h\st_p(\Phi)$ and the tensor flip map
 \[T: X\rra Y\ \ , \ T(\xi\ten h)=h\ten \xi \pl  \ \ \text{for}\pl  \xi\in \st_p(\Phi),\pl h\in l_2^n\pl .\]
 According to \cite{pvp}, we know  that
 \[ \|id_{C_p^n}\ten T: C_p^n\ten_h X\to C_p^n\ten_h Y\|
 \lel \|id_{R_n}\ten T: R_n \ten_h X\to R_n\ten_h Y\| \pl .\]
Moreover, using the little Fubini theorem \eqref{hcp} $C^n_p\ten_h C_p=C_p\ten_h C_p^n$ we see that
 \begin{align*} \crp_p(\st_p(id_n\ten \Phi))
 &\lel \|id_{R\to C^p}\ten id_{\st_p(id_n\ten \Phi)}: [C_p^n\ten_h \st_p(\Phi)]\ten_h R\to
 C_p\ten_h [C_p^n\ten_h \st_p(\Phi)]\|\\
 &\lel \|id_{C^n_p}\ten T: C_p^n\ten_h X\to C_p^n\ten_h Y\| \pl .\end{align*}
Then the first step we recall that the tensor flip map from
$R_n\ten_h X\to X\ten_h R_n$ is a contraction. Indeed, we have
 \[ R_n\ten_h X \subset R_n\ten_{\min}X \cong X\ten_{\min}R_n \lel X\ten_{h} R_n \pl.\]
The inclusion is completely contractive since the minimal tensor product is the smallest operator space tensor product norm \cite{Psbook}. Then we see that
 \[ \|id\ten T: R_n \ten_h \st_p(\Phi)\ten_h R\to C_p^n\ten_h C_p \ten_h \st_p(\Phi)\|\le  \crp_p(\Phi) \pl .\]
Finally, we have to replace $C_p^n$ by $R_n$ and use the fact that $\|id:C_p^n\to R_n\|_{cb}=n^{1/2p'}$,  which can be easily proved by interpolation. This implies
 \[ \|id_{R_n}\ten id_{X\to Y}: R_n \ten_h X\to R_n\ten_h Y\|\kl n^{1/2p'}\crp_p(\st_p(\Phi)) \]
and concludes the proof of the upper bound. The equality follows from tensor norm property
 \[ \|\sum_j x_j\ten x \ten y_j\|_{C_p^n\ten_h X\ten_h  R} \lel \|\sum_j x_j\ten y_j\|_{C_p^n\ten_h R} \|x\|_{X}  \pl,\]
which could be easily verified using the definition of Haagerup tensor product. \qd

The center of our analysis is a special class of completely positive and trace preserving maps, which in operator algebra literature are called conditional expectations. Let us recall the definition and some basic properties. (See again \cite{Tak} for a reference). For an inclusion $M\subset (N,tr)$ of semi-finite von Neumann algebras such that $tr|_{M}$ is still a semi-finite trace ($M$ admits enough positive elements with ${tr(x)<\infty}$),  the \emph{conditional expectation} from $M$ to $N$ is the unique completely positive unital and trace preserving map
$\E_M:N\to M$ such that
\begin{align}\label{proj} tr(\E(x)y) \lel tr(xy) \quad \mbox{for}\quad  x\in N\pl,\pl  y\in M \pl.\end{align}
%The bi-module property
%\begin{align*}\E_M(axb)=a\E_M(x)b,\  \ \ \forall x\in N, a,b\in M\end{align*}
%will also be used.
%The definition \eqref{proj} admits an immediate quantum mechanical interpretation. Suppose a state $\rho$  is given for $N$, but measurements can only be performed
%with respect to observables in $M$. Then $\E(\rho)$ is the state observed by $M$, the algebra of restricted measurements.
In finite dimension we encounter several equivalent descriptions. We will assume that $M\subset M_m$ and $M'\subset M_m$ is the commutator. Then the unitary group $U(M')$ of $M'$ is a compact group and admits a Haar measure $\mu$. Let us consider the averaging map of unitary conjugation
\begin{align} \Phi(x)\lel \int_{U(M')} u^*xu d\mu(u) \pl\pl \text{for} \pl x\in M_m\pl. \label{phie}\end{align}
Certainly for all $y\in M$, $\Phi(y)=y$ and
  \[ tr(\Phi(x)y) \lel \int_{U(M')} tr(u^*xuy) d\mu(u)
  \lel \int_{U(M')} tr(xuyu^*) d\mu(u) \lel tr(xy) \pl.\]
Then by the definition \eqref{proj}, $E_M=\Phi$. Moreover, we see that $\E$ also defines a contraction on the space $L_2(M_m,tr)=S_2^m$, the matrix space equipped with Hilbert-Schmidt norm. Actually $\E$ is the unique orthogonal projection from $L_2(M_m,tr)$ to the subspace $L_2(M,tr)$ equipped with the induced trace. Recall that finite dimensional $C^*$-algebras are semi-simple and hence  we may assume that $M=\oplus_k M_{n_k}$ is a direct sum of matrix algebras. The projection $P_k\in M_m$ onto the each blocks $M_{n_k}$ are mutually orthogonal and form a von Neumann measurement. Moreover, the embedding of ${M_{n_k}\subset P_kM_mP_k=M_{n_km_k}}$ has a certain multiplicity $m_k$. This means the inclusion $M\subset M_m$ is given by
 \[ M\cong \oplus_k( M_{n_k}\ten 1_{m_k}) \subset M_m \pl .\]
The induced trace has to be given by $tr((x_k)_k)=tr(\oplus_k (x_k\ten 1_{M_{m_k}}))=\sum_k m_ktr(x_k)$. Then the conditional expectation has a concrete expression $\E_M=\oplus_k (id_{n_k}\ten tr_{m_k})$. In other words, the conditional expectation is always a direct sum of partial traces, depending on the matrix block and multiplicity of $M$. Let us introduce the following notation: for a finite dimensional von Neumann algebra $M\cong \oplus_k M_{n_k}$, we denote
\[d_M=\text{the size of the largest diagonal block}=\max_k n_k\pl.\]
By the $Q^{(1)}$ formula of direct sum channels in \cite{Wolf}, it is immediate to see that for any conditional expectation $\E_M: M_m \to M$,
\[Q^{(1)}(\E_M)=Q(\E_M)=Q^{(p)}(M)=\ln d_M\pl.\]
Here we reprove the above statement by calculating the row-column $p$-concavity.

\begin{prop}\label{apropiori}  Let $M=\oplus_k (M_{n_k}\otimes 1_{M_{m_k}})\subset M_m$ be a von Neumann subalgebra, and ${\E_M:M_m\to M}$ be the conditional expectation. Then
\[Q^{(1)}_p(\E_M)=d_M^{1/p'}\pl, \pl Q^{(1)}_p(\E_M\ten \Psi)\lel \pl d_M^{1/p'}Q_p^{(1)}(\Psi)\pl ,\]
for any channel $\Psi$. This implies $Q^{(1)}(\E_M)=Q(\E_M)= Q^{(p)}(\E_M) \lel \ln d_M \pl .$
 \end{prop}
\begin{proof} The first equality follows easily from Example \ref{ex1} and \ref{ex2}. Now we consider an additional channel $\Psi:S_1(H_{A''})\to S_1(H_{B''})$. Then $\E_M\ten \Psi$ is still block-diagonal, and hence we can combine Example \ref{ex2} and \ref{ex3} to deduce that
 \begin{align*} \crp_p(\st_p(\E_M\ten \Psi)) &\lel \max_k\pl \crp_p(\st_p(id_{n_k}\ten \tilde{tr}_{m_k}\ten \Psi))\\
 &\lel \max_k \pl\crp_p(\st_p(id_{n_k}\ten \Psi))
 \lel d_M^{1/2p'}\crp_p(\st_p(\Psi) )\pl .\end{align*}
Here we used that the output state can be changed via an isometry in the  Stinespring space. By Proposition \ref{rcq1}, we have  \[Q_p^{(1)}(\E_M\ten \Psi)=d_M^{1/p'}Q_p^{(1)}(\Psi)\ ,\  Q^{(1)}(\E_M\ten \Psi)= \ln d_M+Q^{(1)}(\Psi)\pl,\]
which completes the proof.\qd

%% file: Qcap_comparisonthm.tex
\section{The Comparison Theorem}
\subsection{The standard form of a von Neumann algebra}
Let $M$ be a von Neumann algebra equipped with a normal faithful trace $tr$,
the GNS construction with respect to the trace $tr$ consists of the Hilbert space $L_2(M,tr)$ obtained of the completion of $M$
with respect to the norm
$\|x\|_2=tr(x^*x)^{1/2}$. The symbol  ``$tr$'' in $L_2(M,tr)$ will be frequently omitted if it is clear from the context.
We will always distinguish operators $x\in M$ from their corresponding vectors $|x\ran\in L_2(M,tr)$. If $tr$ is faithful and $M$ is finite dimensional, then  $L_2(M)$ and $M$ are really the same set.  The distinction is nevertheless meaningful, and necessary in infinite dimension. We will denote the GNS representation of a normal faithful trace by $\la$, namely
\begin{align*}\la : M\rra B(L_2(M, tr))\quad, \quad
 \la(x)|y\ran \lel x|y\ran \lel |xy\ran \pl .\end{align*}
Note that $\la$ is injective since $tr$ is faithful. We will also frequently omit ``$\la$'' and simply write $``x|y\ran''$. A key part of the GNS-construction is the anti-linear isometric involution $J_M(|x\ran)=|x^*\ran$ which relates $M$ and its commutant $M'$ in $B(L_2(M))$
 \begin{align*} J_M\la(M)J_M \lel \la(M)'=\{T \in B(L_2(M,tr)\pl|\pl \forall {x\in M}\ ,\pl T\la(x)=\la(x)T \} \pl .\end{align*}
Indeed, let us observe that
 \begin{align} J_Mx^*J_My|z\ran \lel J_Mx^*J_M|yz\ran
 \lel J_M|x^*(z^*y^*)\ran
 \lel |yzx\ran \lel yJ_Mx^*J_M|z\ran \pl .\label{commu}\end{align}
In other words the inclusion $J_MMJ_M\subset M'$ is trivial. The converse inclusion can be found in any standard reference on operator algebra (e.g. \cite{Tak}). The formula
 \begin{equation} \label{xy}
  J_My^*J_M|x\ran \lel |xy\ran \lel x|y\ran
  \end{equation}
will be frequently used. We extend the bracket notation from $M$ to $B(L_2(M))$ as follows
\begin{align*}\iota : B(L_2(M))\to L_2(M)\pl,\ \ \iota(x)=x\ket{1}:=|x\ran \pl ,\end{align*}
and also its dual version
\begin{align*}\bar{\iota} : B(L_2(M))\rra L_2(M)^*,\ \ \bar{\iota}(x)=\bra{1}x:=\bra{x^*} \pl .\end{align*}
In particular, for $x'=J_Mx^*J_M\in M'$ we obtain
$\displaystyle |x'\ran \lel J_Mx^*J_M|1\ran \lel |x\ran\pl.$
\begin{exam}\label{GNS}{\rm The most elementary example is $(M_n,tr)$, the matrix algebra and its full trace $tr(1)=n$. Its GNS construction gives a natural embedding of $M_n$ into $M_n\otimes M_n$ satisfying
\begin{align*}
L_2(M_n, tr)\cong l_2^n\otimes_2 l_2^n=S_2^n &\ \ \ \ \lambda: M_n \ \rra B(l_2^n\otimes l_2^n)\cong M_n\otimes M_n\\
\ket{e_{ij}} \rra e_i\otimes e_j\ ,\ \ \ \ \ \ \ \ \ \ \ \ \ &\ \ \ \ \ \ \lambda(a)=a\otimes 1.
\end{align*}Here $S_2^n$ is the matrix space equipped with the Hilbert-Schmidt norm.
The operator $J$ in this case is
\[J(e_i\ten e_j)=J\ket{e_{ij}}=\ket{e_{ji}}=e_j\ten e_i,\ \ \ J(a\ten 1)J=1\ten \bar{a}\pl,\]
where $\bar{a}$ is the entry-wise complex conjugation of matrix $a$.
}\end{exam}
Let us recall Haagerup's definition of the standard form of a von Neumann algebra.
\begin{defi} Given a von Neumann algebra $M\subset B(H)$, a quadruple $\{M,H,J,H_+\}$ given by a unitary involution $J$, a self-dual cone $H_+$ in $H$ is said to be a standard form for $M$ if $\text{\rm i)}\ JMJ=M';\ \text{\rm ii)}\ JaJ=a^*,\ {a\in M\cap M'} ;\ \text{\rm iii)}{\ Jh=h, \pl h \in H_+;}$  {\rm iv)} ${aJaJH_+\subset H_+,\  a\in M.}$
 \end{defi}

For finite dimensional $M$ with a faithful trace $tr$, $(M, L_2(M, tr),J_M,L_2(M_+))$ is \emph{the} canonical standard form of $M$,  since all standard forms of $M$ are unitarily equivalent. We say that an inclusion $M\subset M_m$ is \emph{standard} if it is unitarily equivalent to
GNS representation of the induced trace $tr$. We refer to \cite{Haastandard} and \cite{Tak} for more information about standard forms.

Let $U\in M_m\ten N$ be an unitary and $\theta_f: S_1^m \to S_1^m$ be an VN-channel via
\begin{align}\theta_f(\rho)= id\ten\tau(U(\rho\ten f)U^*)\pl .\label{channel}\end{align}
We consider the following conditions on $N$ and $U$:
\begin{enumerate}
\item[C1)] There exists a standard inclusion $M\subset M_m$ of a $*$-subalgebra $M$;
\item[C2)] $U$ admits a tensor representation $U=\sum_i x_i\ten y_i$ with $x_i\in M'$, $y_i\in N$;
\item[C3)] The operator $B=\sum_i |x_i\ran\ten \lan y_i^*|\in B(L_2(N,\tau), L_2(M,tr))$ satisfies $BB^*= id_{L_2(M)}$;
\item[C4)] There exists a scalar $\mu>0$ such that $B^*B=\mu id_{L_2(N)}$.
\end{enumerate}

Choosing a basis in $M'\cong M$, we may then always write every element $U\in M'\ten N$ as $U=\sum_i x_i\ten y_i$ with $x_i\in M'$, $y_i\in N$. Hence the operator $B$ is uniquely determined by $U$. Using these operators we find an even more explicit form of a VN-channel
 \begin{align} \theta_f(\rho)\lel \sum_{i,j} \tau(y_ify_j^*) x_i\rho x_j^* \pl \label{channelmap2}.\end{align}
By unitary equivalence of standard forms, we may and will assume that $\theta_f$ is from $S_1(L_2(M))$ to itself, namely $H_{A'}=H_B=L_2(M)$. The following lemma characterizes the Stinespring space of $\theta_f$.

\begin{lemma}\label{eleprop} Assume {\rm C1), C2)} and {\rm C3)}.
Let $f$ be a density and $\theta_f$ be the corresponding VN-channel. Let $V_f\in B(L_2(M),L_2(M)\otimes L_2(N))$ be defined by  $V_f(h)=\sum_i |x_i(h)\ran \ten \ket{ y_i\sqrt{f}}$. Then
 \begin{enumerate}
 \item[i)] $V_f$ is the partial isometry of $\theta_f$ such that $V_f^*V_f=id_{L_2(M)}$ and
 \begin{align*}\theta_f(\rho)=id\otimes tr (V_f\rho V_f^*)\pl ;\end{align*}
\item[ii)]The Stinespring space of $\theta_f$ is given by
   \begin{align*}\st(\theta_f)=\{V_f(h)| h\in L_2(M)\}=(M\ten J_N\sqrt{f}J_N)(\sum_i |x_i\ran\ten |y_i\ran)\pl;\end{align*}
\item[iii)] Let $\sigma: L_2(N)\to L_2(N)^*$ be the isometry  given by $\sigma(\ket{a})=\bra{a^*}$. Then \[(id\ten \sigma)\st(\theta_f)=MB\sqrt{f}\pl .\]
\end{enumerate}
\end{lemma}
\begin{proof} We will denote full traces of $B(L_2(M))$ and $B(L_2(N))$ as ``$tr$''. For i), we start with the second identity. Indeed, using the fact that $\tau$ is a trace we find for $h,k\in M$
 \begin{align*}
   \theta_f(|h\ran \lan k|)&= \sum_{i,j}\tau(\sqrt{f}y_j^*y_i\sqrt{f}) |x_ih\ran \lan x_jk|  \lel
   \sum_{i,j} tr(|y_i\sqrt{f}\ran \lan y_j\sqrt{f}|) |x_ih\ran \lan x_jk|\\
   &=  id\ten tr(|V_f(h)\ran \lan V_f(k)|) \pl .
   \end{align*}
Since $\theta_f$ is obviously trace preserving, we deduce that $V_f$ is a partial isometry by taking traces. Indeed,
 \[ \lan V_f(h)|V_f(k)\ran \lel tr\ten tr(|V_f(h)\ran \lan V_f(k)|) \lel tr(\theta_f(|h\ran \lan k|))=\lan h|k\ran \pl.\]
The first equality of ii) follows from i). Now choose $x_i'\in M$ such that $x_i=J(x_i')^*J\in M'$,
 \begin{align*} x_i|h\ran=J(x_i')^*J|h\ran=|hx_i'\ran=h|x_i'\ran=h|x_i\ran \pl .\end{align*}
Together with $J_N\sqrt{f}J_N(|y_i\ran)=|y_i\sqrt{f}\ran$ this proves ii). Moreover, iii) follows from that for $\ket{h}\in L_2(M)$
 \begin{align*}(id\ten \sigma)(V_f\ket{h})=(id\ten \sigma) (\sum_i |x_i(h)\ran \ket{y_i\sqrt{f}})=\sum_i |x_i(h)\ran \lan \sqrt{f}y_i^*|= hB\sqrt{f}\pl. \pll\pll\pll\pll\pll 
 \qedhere   \end{align*}
\end{proof}

\subsection{Proof of Theorem \ref{comp}} The proof of the Comparison Theorem is divided into several pieces. Our first observation is based on the different descriptions of conditional expectations.
\begin{lemma}\label{oproj} \label{ortoproj} Let $H,K$ be finite dimensional Hilbert spaces. Let $M\subset B(H)$ be a $*$-subalgebra. Then
\begin{enumerate}
\item[i)] the conditional expectation $\E_M$ is completely contractive from $H^{c_p}\ten_h H^r$ onto $M$ for all $1\le p\le \infty$;
\item[ii)] let $B\in B(H,K)$ be a partial isometry such that $BB^*=id_K$. Then the orthogonal projection from $H\ten_2 K$ onto $MB$ is a complete contraction on $H^{c_p}\ten_h K^r$ for all ${1\le p\le \infty}$.
\end{enumerate}
\end{lemma}
\begin{proof} The conditional expectation  $\E_M:B(H)\to M$ is completely positive and unital, and  hence completely contractive on $B(H)=H^c\ten_h H^r$. According to \eqref{phie}, we know that $\E_M$ is also a contraction, and by homogeneity of $(H^r\ten_h H^r)=(H\ten_2H)^r$ even a complete contraction for $p=1$.  Then the first assertion follows from interpolation \[{H^{c_p}\ten_h H^r=[H^c\ten_h H^r,H^r\ten_h H^r]_{1/p}}\pl .\] For the second assertion we observe that
the orthogonal projection $P_{MB}$ from $H\ten_2 K$ onto $MB$ can be factorized as $P_{MB}(T)=\E_M(TB^*)B$. Indeed, $T\mapsto \E_M(TB^*)B$ is contractive and satisfies  $\E_M(yBB^*)B=yB$ for $y\in M$. By uniqueness of the orthogonal projection we get $P_{MB}(\cdot)=\E_M(\cdot \pl B^*)B$. Since $P_{MB}$ is an orthogonal projection, it is completely contractive on $H^r\ten_h K^r$ (when $p=1$). For $p=\infty$ we note that right multiplication $R_a(x)=xa$
is completely contractive for any contraction $a$. In particular, ${P_{MB}=R_B\circ\mathcal{E}\circ R_{B^*}}$ is completely contractive on $H^{c}\ten_h K^r$. Again interpolation yields the assertion.
\qd

In Lemma \ref{eleprop}, we calculated the Stinespring spaces of $\theta_f$ for a given density $f$. We may formally extend the definition for arbitrary  $a\in N$ as follows
 \begin{align*} \st(a)=U(L_2(M)\ten\ket{a}) \lel \{\sum_i |x_i(h)\rangle | y_ia\ran \pl|\pl h\in M \}
  \subset L_2(M)\ten L_2(N) \pl.
\end{align*}
If we want to emphasize the operator space structure, we denote
 \begin{align*}\st_p(a)\lel MBa \lel \{\sum_i |x_i(h)\rangle \lan a^*y_i^*|  \pl|\pl h\in M \} \subset   L_2^{c_p}(M)\ten_h L^r_2(N) \pl. \end{align*}
\begin{lemma}\label{spemap} Assume {\rm C1), C2)} and {\rm C3)}.
Let $a_1,a_2$ be unitaries in $N$. Then the map
 \[\Phi_{a_1,a_2} \lel U(id\ten |a_1\ran \lan a_2|)U^*\pl .\]
is a complete contraction on $L_2^{c_p}(M)\ten_h L_2^{r}(N)$ for all $1\le p\le \infty$.
\end{lemma}
\begin{proof} Let us start with $a_1=a_2=1$. Recall that Lemma \ref{eleprop} implies
 \begin{align*} \st(1) =\st(\theta_1)\lel U(L_2(M)\ten |1\ran)  \end{align*}
and hence $\Phi_{1,1}$ is the unique orthogonal projection from the Hilbert space $L_2(M)\ten_2 L_2(N)$ onto $\st(1)$. Moreover, we also know that $\st_p(1) \lel MB$. By Lemma \ref{ortoproj}, $\Phi_{1,1}$ is a complete contraction for all $1\le p\le \infty$. For general $a_1,a_2\in N$ we note that $U(1\ten J_NaJ_N)=(1\ten J_NaJ_N)U$ commutes because $U\in M'\ten N$. This implies
\begin{align*}
  U(1\ten |a_1\ran \lan a_2|)U^*
 &=  U(1\ten J_Na_1^*J_N)(1\ten |1\ran \lan 1|)(1\ten J_Na_2J_N)U^* \\
 &=  (1\ten J_Na_1^*J_N) \Phi_{1,1} (1\ten J_Na_2J_N) \pl .
\end{align*}
By the properties of the Haagerup tensor product (see \cite{Psbook}) we know that the first and the third terms are complete contractions for unitaries $a_1, a_2$. Clearly the composition of three complete contractions is again a complete contraction. \end{proof}

\begin{theorem}\label{compair} Assume {\rm C1), C2)} and {\rm C3)}. Let ${\rho \in S_1(H_A\ten L_2(M))}$ be a bipartite state for some Hilbert space $H_A$ and $f_1,f_2\in L_1(N,\tau)$ be densities . Then for all $1\le p\le \infty$,
\[ \|id_A\ten \theta_{f_1}(\rho)\|_{p}
  \kl \|f_1\|_p\|f_2\|_p \|id_A \ten \theta_{f_2}(\rho)\|_{p} \pl.\]
\end{theorem}
\begin{proof} Fix a $p\in [1,\infty]$, we introduce $\displaystyle a_k=\frac{\sqrt{f_k}}{\|\sqrt{f_k}\|_{2p}}$ for $k=1,2$. We claim that the map
\[ \Phi_{a_1,a_2}\lel U(id\ten |a_1\ran \lan a_2|)U^* \]
is a complete contraction on $L_2^{c_p}(M)\ten_h L_2^{r}(N)$. Indeed, let us first assume that $a_k$ is invertible.  Since $\|a_k\|_{2p}=1$ and $a_k> 0$, we may define the analytic functions $a_k(z) \lel a_k^{pz}$. Thus we obtain an analytic family of maps
\[\Phi(z) \lel U(id\ten |a_1(z)\ran \lan a_2(\bar{z})|)U^* \pl .\]
For $z=it$, $a_1(it)$ and $a_2(-it)$ are unitaries. Hence by Proposition \ref{spemap}, $\Phi(it)$ is a complete contraction on $L_2^{c}(M)\ten_hL_2^{r}(N)$. For $z=1+it$ we see that $\|a_k^{p(1+it)}\|_2=\tau(a_k^{2p})^{1/2}=1$ for $k=1,2$. Then $\Phi(1+it)$ is a partial isometry on $L_2^r(M)\ten_h L_2^r(N)$. By Theorem \ref{stein}
(Stein's interpolation theorem), we deduce for $z=1/p$ that
\begin{align*} \|\Phi(1/p):L_2^{c_p}(M)\ten_h L_2^{r}(N)\to
 L_2^{c_p}(M)\ten_h L_2^{r}(N)\|_{cb} \kl 1 \pl .\end{align*}
For $h\in L_2(M)$, denote $\eta \lel U(h\ten \ket{\sqrt{f_2}})$,  we have
 \begin{align*} \Phi(1/p)(\eta) \lel \lan a_2,\sqrt{f_2}\ran
 U(h\ten |a_1\ran)
  &\lel \frac{\lan a_2,\sqrt{f_2}\ran}{\|\sqrt{f}_1\|_{2p}}   U(h\ten |\sqrt{f_1}\ran) \\&\lel \frac{1}{\|\sqrt{f}_1\|_{2p}\|\sqrt{f}_2\|_{2p}}   U(h\ten |\sqrt{f_1}\ran)\pl .\end{align*}
Therefore the ``transition map'' between the Stinespring spaces
${T_{f_1,f_2}: \st_p(\theta_{f_2}) \to \st_p(\theta_{f_1})}$ defined by
 \begin{align*} T_{f_1,f_2}(\sum_i |x_ih\ran \ten \lan \sqrt{f_2}y_i^*|)
 \lel (\sum_i |x_ih\ran \ten \lan \sqrt{f_1}y_i^*|) \end{align*}
satisfies
 \begin{align*} \|T_{f_1,f_2}:\st_p(\theta_{f_2})\to  \st_p(\theta_{f_1})\|_{cb}
   \kl \|\sqrt{f_1}\|_{2p}\|\sqrt{f_2}\|_{2p} \pl .\end{align*}
Applying this to an element $\xi\in B(H_A\ten L_2(M))$, we deduce from Lemma \ref{shuffle} that
 \begin{align*}
  &\|id_A\ten \theta_{f_1} (\xi \xi^*)\|_{S_p(H_A\ten L_2(M))}
  \lel  \norm{\sum_i (1_A\ten V_{f_1})\xi}{H_A^{c_p}\ten_h \st_p(\theta_{f_1})\otimes_h L_2(M)^r\otimes_h H_A^{r}} ^2\\=&\norm{\sum_i (1_A\ten x_i)\xi\otimes \bra{\sqrt{f_1}y_i^*}}{H_A^{c_p}\ten_h \st_p(\theta_{f_1})\otimes_h L_2(M)^r\otimes_h H_A^{r}} ^2\\
  &\kl  \|\sqrt{f_1}\|_{2p}^2 \|\sqrt{f}_2\|_{2p}^2\norm{\sum_i (1_A\ten x_i)\xi\otimes \bra{\sqrt{f_2}y_i^*}}{H_A^{c_p}\ten_h \st_p(\theta_{f_2})\otimes_h L_2(M)^r\otimes_h H_A^{r}} ^2\\
  &= \|\sqrt{f_1}\|_{2p}^2 \|\sqrt{f}_2\|_{2p}^2
  \|id_A\ten \theta_{f_2}\ (\xi \xi^*)\|_{S_p(H_A\otimes L_2(M))} \end{align*}
holds for all positive $\rho=\xi\xi^*\in S_1(H_A\ten L_2(M))$. Using
$\|\sqrt{f}_k\|_{2p}^2=\tau(f_k^p)^{1/p}=\|f_k\|_p$ for $k=1,2$ implies the assertion in case of invertible densities $f_1, f_2$. For noninvertible densities we first consider $\delta>0$ and $\tilde{f_k}=f_k+\delta 1$ invertible. The same argument shows that
 \begin{align*} \|T_{\tilde{f_1},\tilde{f_2}}:\st_p(\theta_{\tilde{f_2}})\to \st_p(\theta_{\tilde{f_1}})\|_{cb} \kl  \|\sqrt{\tilde{f_1}}\|_{2p}\|\sqrt{\tilde{f_2}}\|_{2p}  \pl .\end{align*}
The assertion in general follows by sending $\delta\to 0$.   \qd
The second inequality of Theorem \ref{comp} follows from above theorem by choosing $f_2=1$. We prove the the first inequality of Theorem \ref{comp} by the following lifting property.

\begin{lemma}\label{mu}  Assume {\rm C1), C2)} and {\rm C3)}. Then $\theta_1$ is the conditional expectation $\E_M$ from $B(L_2(M))$ onto $M$. Moreover, $\theta_1\theta_f=\theta_1$ for all densities $f\in N$.
\end{lemma}
\begin{proof} It suffices to consider rank one matrices $\ket{k}\bra{h}\in B(L_2(M))$ with $k,h \in M$. Since $x_i\in M'$ we find
\begin{align*}
\theta_1(\ket{k}\bra{h})&=\sum_{i,j}\tau(y_iy_j^*)x_i \ket{k}\bra{h}x_j^* \lel \sum_{i,j}\tau(y_iy_j^*) \ket{x_ik}\bra{x_jh}\\
&=\sum_{i,j}\tau(y_iy_j^*) \ket{kx_i}\bra{hx_j}\lel
  k(\sum_{i,j} \lan y_i^*,y_j^*\ran  \ket{x_i}\bra{x_j})h^*\lel k BB^*h^* \lel kh^*\pl.
\end{align*}
Then  we  observe that for any $a\in M$,
\begin{align*}tr(\ket{k}\bra{h}a) \lel \bra{h}a\ket{k} \lel tr(h^*ak) \lel tr(kh^*a)\pl .\end{align*}
Thus $\theta_1=\E_M$ is the conditional expectation onto $M$ by the definition. For ii), thanks to \eqref{phie} the conditional expectation is given by the integral over $U(M')$. Let $f\in N$ be a density, and $\ket{k}\bra{h}\in B(L_2(M))$ again a matrix unit. Then we have \begin{align*}\E_M[\theta_f(\ket{k}\bra{h})]=\int_{U(M')}ukBfB^*h^*u^*du=k\E_M(BfB^*)h^*\pl \pl. \end{align*}
Thus it suffices  to show $\E_M(BfB^*)=1$. For positive $x\in M$ we have
 \begin{align*}
  tr(xBfB^*) \lel \|\sqrt{x}B\sqrt{f}\|_2^2\pl .
  \end{align*}
Then we note that
 \begin{align*} \sqrt{x}B\sqrt{f}
 \lel \sqrt{x}(\sum_i |x_i\ran \ten \bra{ y_i^*})\sqrt{f}
 \lel  \sum_i |\sqrt{x}x_i\ran \ten \bra{ \sqrt{f}y_i^*} \pl .\end{align*}
Recall that  $\si(|x\ran)=\lan x^*|$ is a linear isometry and thus \begin{align*}
  tr(xBfB^*) &= \|\sqrt{x}B\sqrt{f}\|_2^2 \lel
 \|(1\ten \si)\sqrt{x}B\sqrt{f}\|_2^2 \lel
  \|(\sum_i x_i\ten y_i) (|\sqrt{x}\ran \ten |\sqrt{f}\ran)\|_2^2 \\
 &= \|U(|\sqrt{x}\ran \ten |\sqrt{f}\ran)\|_2^2 \lel
 \|(|\sqrt{x}\ran \ten |\sqrt{f}\ran)\|_2^2
 \lel tr(x) \tau(f) =tr(x)\pl .
 \end{align*}
By linearity this remains true for all $x\in M$, which completes the proof.
\end{proof}

\begin{prop}\label{th1f} Assume {\rm C1), C2)} and {\rm C3)}.  Let ${\rho \in S_1(H_A\ten L_2(M))}$ be a bipartite state with some Hilbert space $H_A$, and $f_1,f_2\in L_1(N,\tau)$ be densities. Then for all $1\le p\le \infty$,
\begin{align*}\|id_A\ten \theta_1(\rho)\|_p \kl \|(id_A\ten \theta_f)(\rho)\|_p \pl .\end{align*}
\end{prop}

\begin{proof} According to Lemma \ref{mu} we have
 \[ (id_A\ten \theta_1) \lel (id_A\ten \theta_1)(id_A\ten \theta_f) \lel (id_A\ten \E_M) (id_A\ten \theta_f)\pl .\]
However, $id_A\ten \E_M$ is a unital and trace preserving completely positive map and hence a contraction on $S_p(H_A\ten L_2(M))$ for all $1\le p\le \infty$.
\qd 

%% file: Qcap_cbentropy.tex
\section{Negative Cb-entropy and Combined bounds}
\subsection{Negative cb-entropy}
\noindent The cb-entropy was first introduced
in \cite{DJKR}, and rediscovered as
``reverse coherent information'' in \cite{invcoh}. We will give a formula of the cb-entropy of $\theta_f$ using condition C4).  The ideas go back to the so far unfortunately unpublished manuscript \cite{JNR2}. Let us recall that for a channel ${\Phi :S_1(H_{A'})\to S_1(H_B)}$, the negative cb-entropy of $\Phi$ is defined as
\begin{align*}
-S_{cb}(\Phi)=\sup_{\rho\ pure}{H(A)_\si-H(AB)_\si}\pl.
\end{align*}
Here $H(A)-H(AB)=I_c(B\ran A)$ motivates the terminology ``reverse coherent information''. Our discussion is based on the differential description  from \cite{DJKR},
\begin{equation}\label{djkr}
-S_{cb}(\Phi) \lel \dep{\norm{\Phi:S_1(H_{A'})\rra S_p(H_B)}{cb}} \pl .
\end{equation}
Using $CB(X,Y)\cong X^*\ten_{\min} Y$, we may consider the vector-valued $(\infty,p)$ norm defined in \eqref{QPP} for its Choi matrix. Indeed, assuming a basis $\{e_i\}_{1\le i\le m}$ for $H_{A'}$, the Choi matrix of $\Phi:S_1(H_{A'})\to S_1(H_B)$ is given by
 \begin{align*} \large{\chi_{\Phi}} \lel \sum_{i,j} e_{i,j}\ten \Phi(e_{i,j})
 \lel m \pl \big(id\ten \Phi(|\psi_m\ran \lan \psi_m|)\big) \pl ,\end{align*}
where $\ket{\psi_m}=\frac{1}{\sqrt{m}}\sum_i e_i\ten e_i $ is a maximally entangled state in $H_{A'}\ten H_{A'}$ with $|A'|=m$. The complete isometry
 \begin{align*} CB(S_1(H_{A'}),S_p(H_B))\cong B(H_{A'})\ten_{\min} S_p(H_B)\lel M_m(S_p(H_B))\pl\end{align*}
is explicitly given by the Choi matrix
\begin{align*}\norm{\Phi:S_1(H_{A'})\rra S_p(H_B)}{cb} \lel \norm{\chi_\Phi}{M_m(S_p(H_B))}\pl .\end{align*}

\begin{theorem} \label{cbentropy} Let $N\subset B(L_2(N))$ be $n$-dimensional von Neumann algebra with induced faithful normalized trace $\tau=\frac{tr}{n}|_N$. If $U=\sum_i x_i\otimes y_i$ is a unitary in $M_m\otimes N$ such that ${\rm B}=\sum_{i}\ket{x_i}\otimes \bra{y_i^*}$ in $B(L_2(N),L_2(M_m))$ satisfies ${\rm B}^*{\rm B}= \mu\pl id_{L_2(N)}$. Then
$\mu=\frac{m}{n}$ and
\begin{align*} -S_{cb}(\theta_f) \lel  \ln \mu +\tau{(f\ln f)}\pl , \end{align*}
where the optimal value is  attained at maximally entangled states.
\end{theorem}

\begin{proof} First, the equality $\mu=\frac{m}{n}$ follows easily from computing the traces,
\begin{align*}m=tr\ten \tau (U^*U) = tr_{B(L_2(N))}({\rm B}^*{\rm B})=tr_{B(L_2(N))}(\mu \pl id_{L_2(N)})=n \mu\pl.\end{align*}
Let $\psi_m$ be a maximally entangled state in $M_m\ten M_m$ and a matrix $a$ be in $M_m$. Then \[(a\ten 1)(m^{1/2}\ket{\psi_m})=\sum_{ij}a_{ij} |i\ran |j\ran=|a\ran\] is the GNS vector of $a$ in $L_2(M_m, tr)$. This implies that
 \begin{align*}
 (\theta_f\ten id)(m|\psi_m\ran \lan\psi_m|)
 &= \sum_{i,j} \tau(y_ify_j^*) [(x_i\ten 1)m|\psi_m\ran \lan\psi_m|(x_j^*\ten 1)]\\
 &\lel \sum_{i,j} \tau(y_ify_j^*) |x_i\ran \lan x_j|
 ={\rm B}f{\rm B}^* \lel \mu\pl wfw^*  \pl ,
 \end{align*}
 where  $w=\mu^{-1/2}{\rm B}$ is a partial isometry satisfying $w^*w=id_{L_2(N)}$. Therefore $\pi(T)=wTw^*$ is a faithful $*$-homomorphism from $B(L_2(N))$ to $B(L_2(M_m))$ and
\begin{align*} tr_{B(L_2(N))}(T) \lel tr_{B(L_2(M_m))}(\pi(T)) \end{align*}
holds for all $T\in B(L_2(N))$. By our assumption $n=\dim N$ and $\tau(T)=n^{-1} tr_{B(L_2(N))}(T)$ for $T\in N$, this implies
\begin{align*}  \|f\|_{L_p(N,\tau)}^p \lel n^{-1}tr_n(|f|^p) \lel
 n^{-1}tr_{m^2}(\pi(|f|^p)) \lel n^{-1}\|w^*fw\|_{S_p^{m^2}}^p \pl .\end{align*}
Therefore we get
 \begin{align*}
   &\|f\|_{L_p(N,\tau)} \lel n^{-1/p}\|w^*fw\|_p
   \lel \mu^{-1}n^{-1/p}\|{\rm B}f{\rm B}^*\|_{p}=
 \mu^{-1}n^{-1/p}
    \|(\theta_f\ten id)(m|\psi_m\ran \lan\psi_m|)\|_{p}\\
 &=
    \mu^{-1}n^{-1/p}
   \|(id\ten \theta_f)(m|\psi_m\ran \lan\psi_m|)\|_p
    \lel \mu^{-1} n^{-1/p} \|\chi_{\theta_f}\|_{p} \pl .
   \end{align*}
For the fourth equality we use that the tensor flip map \[{\rm \flip}(T\ten S)=S\ten T\] is a trace preserving $^*$-homomorphism. In particular, for $p=\infty$ we have
 \begin{align}\label{infty} \|\chi_{\theta_f}\|_{M_m(M_m)} \lel \mu \|f\|_{L_{\infty}(N)} \pl .\end{align}
Moreover, by the definition \eqref{QPP}, we have a lower bound for $M_m(S^m_p)$ norm,
 \begin{equation} \label{lower}
  \|\chi_{\theta_f}\|_{M_m(S_p^m)}\gl m^{-1/p}
 \|\chi_{\theta_f}\|_{S_p^{m^2}}
 \lel \mu n^{1/p}m^{-1/p} \|f\|_p \lel \mu^{1-1/p} \|f\|_p\pl .
 \end{equation}
For the upper bound, we use interpolation. Consider the channel map $\Theta(f)=\chi_{\theta_f}$, by \eqref{infty} it satisfies
 \begin{align*} \|\Theta: L_{\infty}(N)\to M_m(M_m)\|\kl \mu \pl. \end{align*}
On the other hand, for any $H_{A}$ and $\rho\in S_1(H_A\ten H_{A'})$
 \begin{equation}%\label{cpp}
 \begin{split}
  \|(id_A \ten \theta_{f})(\rho)\|_{S_1(H_A\ten H_{B})}
 &= \|(id_{AB}\ten \tau)(1_A\ten U(\rho \ten f)1_A\ten U^*)\|_{S_1(H_A\ten H_{B})} \\&\le \|\rho \ten f\|_{S_1(H_A\ten H_{A'})\hat{\ten} L_1(N)}\lel \norm{\rho}{1}\norm{f}{1}
 \pl . \end{split}
 \end{equation}
This implies for arbitrary $f\in N$
\begin{align*}\norm{\chi_{\theta_f}}{M_m(S_1^m)}=\norm{\theta_f: S_1^m \to S_1^m}{cb}\le \norm{f}{1}\pl, \end{align*}
  and hence $\|\Theta:L_1(N)\to M_m(S_1^m)\|\le 1$. By interpolation \eqref{interpolation}, we deduce that
  \begin{align}\label{upth} \|\Theta:L_p(N)\to M_m(S_p^m)\|\kl \mu^{1-1/p} \pl .\end{align}
Combining \eqref{upth} with \eqref{lower}, the upper and lower bound coincide
 \begin{align*}\|\chi_{\theta_f}\|_{M_m(S_p)} \lel  \mu^{1-1/p}\|f\|_p \pl .\end{align*}
Differentiation \eqref{djkr} implies the formula for
$-S_{cb}(\theta_f)$. Since we used a maximally entangled state $\psi_m$ for the lower bound, this concludes the proof. \qd

\begin{rem}\label{lambda} {\rm In our previous setting we considered $B = \sum_{i=1}^n |x_i\ran_{L_2(M)} \ten \lan y_i^*|$, where we use the right action of $M'$ on $L_2(M,tr)$. These two operators $B$ and ${\rm B}$ are actually related by a partial isometry. Assume $x=J_Mx'^*J_M$ for some $x'\in M$, consider the map
\[W: L_2(M) \to L_2(M_m)\pl , \pl |x\ran_{L_2(M)}=\ket{x'}_{L_2(M)} \to \ket{x}_{L_2(M_m)}\pl.\]
This is well-defined because $M\cong JM'J\subset M_m$ as a standard form. We can choose the specific orthogonal basis $\{\ket{h_i}\}\subset L_2(M) \cong l_2^m$ which satisfies $\sum_j h_jh_j^*=1$. Then for any $x, y\in M'$,
\begin{align*}
\bra{y}x\ran_{L_2(M_m, tr)}&=tr(y^*x)=\sum_i \bra{h_i}y^*x\ket{h_i}=\sum_i \bra{h_i}J_My'x'^*J_M\ket{h_i}\\&=\sum_i \bra{h_i}h_ix'y'^*\ran=tr(\sum_i h_i^*h_ix'y'^*)= tr(x'y'^*)=\bra{y'}x'\ran_{L_2(M)}\pl ,
\end{align*}
Thus $WB\lel {\rm B} $. Of course, this does not change ${\rm B}^*{\rm B}=B^*B$, and hence we may combine Theorem \ref{cbentropy} with Theorem \ref{comp}.
}
\end{rem}

%\begin{rem}{\rm   Note that the $B$ here is slightly different

%with Lemma \ref{apropiori}. It is still independent of tensor representation since $B=(\iota_{M_m}\ten \bar{\iota}_N) U$.
%}\end{rem}
We first have a hashing bound by maximally entangled states.
\begin{prop}\label{lowerbound} Under the assumption of Theorem \ref{cbentropy}, let $\theta_f(1)=\omega_f$. Then $\fs{1}{m}\omega_f$ is a density in $M_m$, and
\begin{enumerate}
  \item[i)] $-S_{cb}(\theta_f)+H(\frac1m \om_f)\le C_{EA}(\theta_f)\le -S_{cb}(\theta_f)+\ln m \pl,$
 \item[ii)]  $-S_{cb}(\theta_f)+H(\frac1m \omega_f)-\ln m\le Q^{(1)}(\theta_f)\le Q(\theta_f)\le \frac12 (-S_{cb}(\theta_f)+\ln m )\pl.$
 \end{enumerate}
In particular, if $\theta_f$ is unital, then
 \begin{align*} -S_{cb}(\theta_f)+\ln m = C_{EA}(\theta_f),\ \
&-S_{cb}(\theta_f)\le Q^{(1)}(\theta_f)\le Q(\theta_f)\le \frac12 (-S_{cb}(\theta_f)+\ln m )\pl. \end{align*}
\end{prop}
\begin{proof}
In the proof of Theorem \ref{cbentropy} we have seen
that $-S_{cb}(\theta_f)$ is attained at a maximally entangled state. This implies
\begin{align*} Q^{(1)}(\theta_f)\gl H(B)-H(A)+(H(A)-H(AB))
 \lel H(\frac1m \omega_f)-\ln m+ (-S_{cb}(\theta_f) )\pl .\end{align*}
The estimate $C_{EA}(\theta_f)\le -S_{cb}(\theta_f)+ \ln m $ follows from ${H(A)=H(A')\le \ln |A|}$ for pure inputs $\rho^{AA'}$. For the lower bound, we see that $C_{EA}(\theta_f)\gl -S_{cb}(\theta_f)+H(\frac1m \omega_f)$ by a maximally entangled input. Moreover, since
$Q\le Q_{EA}=\frac12 C_{EA}$, we deduce the second upper bound for $Q(\theta_f)$.  If $\theta_f$ is unital, $H(\frac1m \omega_f)=H(\frac1m 1)=\ln m$.
\qd
\begin{rem} {\rm Under the assumptions of the Theorem \ref{comp}, we can show that $\theta_f(1)=\E_{M'}(BfB^*)$. Indeed, since the inclusion $M=\oplus_{k=1}^{d} (M_{n_k}\ten 1_{M_{n_k}}) \subset B(L_2(M))$ is standard, we can find an orthonormal basis $\{\frac{1}{n_k}e^k_{rs}|\pl 1\le r,s\le n_k,\pl 1\le k\le d\}$ where the index set has $m=\sum_n n_k^2$ many elements. Denote this basis by $\{
\ket{h_j}|1\le j\le m\}$. For any orthonomal basis we have
$\sum_j \ketbra{h_j}=1$. Thus we get
  \[ \omega_f\lel \theta_f(1) \lel
  \theta_f(\sum_j|h_j\ran\lan h_j|) \lel \sum_j h_jBfB^*h_j^*  \pl .\]
However, for any unitary $u\in M$, $\{\ket{h_ju}\}_{1\le j \le m}$ is also an orthonomal basis and hence, as above, we get
 \[ \omega_f \lel \sum_j h_ju(BfB^*)u^*h_j^* \pl. \]
Averaging over the Haar measure on $U(M)$, we obtain
\begin{align*}
\omega_f &= \sum_j \int_{U(M)} h_ju(BfB^*)u^*h_j^* du
\lel \sum_j h_j\E_{M'}(BfB^*)h_j^*  \\
&=   \E_{M'}(BfB^*) \sum_j h_jh_j^*  \lel \E_{M'}(BfB^*) \pl .
\end{align*}
Here we used that the specific basis satisfies $\sum_j h_jh_j^*=1$ again. Let us recall that C1)-C3) implies $\E_{M}(BfB^*)=1$ for densities $f$, but not necessarily true for $\E_{M'}(BfB^*)$. Actually, a nonunital example is provided in Section 8.}
\end{rem}
Now we are ready to summarize the estimates for quantum capacity. We combine the condition C3) and C4) to be condition C3$'$) as below.
\begin{theorem}\label{best} Let $N\subset B(L_2(N))$ be a von Neumann algebra with induced normalized trace $\tau$. Let $U$ be a unitary in $M_m\otimes N$. For a density $f \in N$, the VN-channel $\theta_f:S_1^m\rra S_1^m$ is given by
\begin{align*}
\theta_f(\rho)=id\otimes \tau(U(\rho\otimes f )U^*).
\end{align*}
Assume that
\begin{enumerate}
\item[C1)] there exist a subalgebra $M\subset M_m$ as a standard inclusion;
\item[C2)] the unitary $U$ admits a tensor representation $U=\sum_i x_i\ten y_i\in M'\ten N$ with ${x_i\in M', y_i\in N}$;
\item[C3$'$)] the operator $B=\sum_i |x_i\ran\ten \lan y_i^*|\in B(L_2(N), L_2(M))$ is a unitary, i.e. ${BB^*= \pl  id_{L_2(M)}}$ and $B^*B= \pl  id_{L_2(N)}$.
\end{enumerate}
Let $M=\oplus_k (M_{n_k}\otimes 1_{M_{n_k}})\subset M_m$ and $\omega_f=\theta_f(1)$. Then
 \begin{enumerate}
\item[i)] $-S_{cb}(\theta_f)=\tau(f\ln f)$;
\item[ii)] $\tau(f\ln f)+H(\frac{1}{m}\omega_f)\le C_{EA}(\theta_f)=2Q_{EA}(\theta_f)\le \ln m+\tau(f\ln f)$;
\item[iii)] $Q(\theta_f)\le Q_{EA}(\theta_f)\le
\frac12(\ln m+\tau(f\ln f))$ and
\[ \max \{\ln d_M, H(\frac1m f)-H(\frac{1}{m}\omega_f)\}\!\le\!
  Q^{(1)}(\theta_f)\le Q(\theta_f)\!\le \! Q^{(pot)}(\theta_f)\!\le\!
  \tau(f\ln f)+ \ln d_M.\]
 \end{enumerate}
\end{theorem}
\begin{proof} Note that $\dim N=\dim M=m$ follows from the assumption C3$'$).
Then combine Corollary \ref{cbounds}, Proposition \ref{apropiori}, Theorem \ref{cbentropy} and Corollary \ref{lowerbound} with fact ${Q\le Q_{EA}=\frac{1}{2}C_{EA}}$. \end{proof}

\begin{rem}{\rm To compare the two upper bounds of $Q(\theta_f)$, we denote by  $\delta=\frac12\ln m-\ln d_M$ the representation gap. If we have $\tau(f\ln f)<2\delta$, then $\tau(f\ln f)+\ln d_M< \frac12{\ln n +\tau(f\ln f)}$, then the comparison bound is better. Otherwise, the entanglement-assisted quantum capacity $Q_{EA}$ gives a better upper bound. We will find examples where $\delta=0$, and hence the comparison property leads to worse bounds for $Q$, but the majorization of $Q^{(p)}$ is not trivial in any case.}
\end{rem}
\begin{rem}{\rm If in addition $\theta_f$ is unital, then the estimates becomes
\begin{enumerate}
\item[i)]
 $\max \{\ln d_M, \tau(f\ln f)\}\!\le\! Q^{(1)}(\theta_f)\le Q(\theta_f)\!\le\! Q^{(pot)}(\theta_f)\!\le\! \tau(f\ln f)+ \ln d_M\pl ;$
\item[ii)] $-S_{cb}(\theta_f)=\tau(f\ln f),\ C_{EA}(\theta_f)=2Q_{EA}(\theta_f)= \ln m +\tau(f\ln f)\pl .$
\end{enumerate}
The Figure.1 gives an illustration of this case.}
\end{rem}
\begin{figure}
\centering
\includegraphics[width=0.5\textwidth]{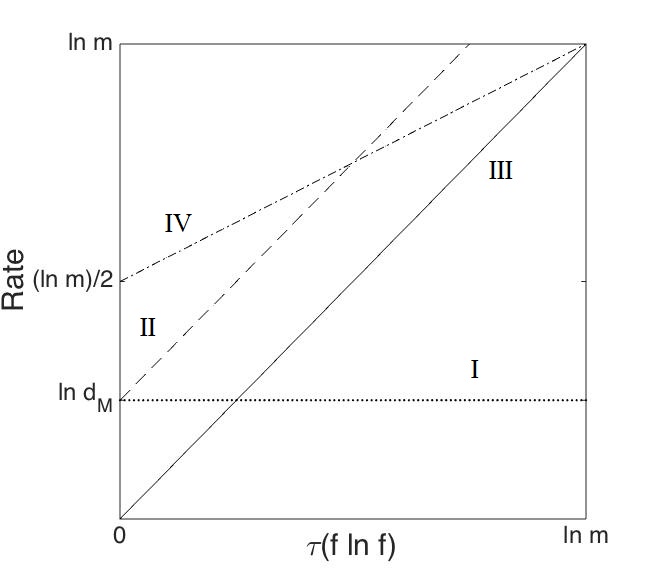}
\caption{Combined bounds for quantum capacity of $\theta_f$ depending on $\tau(f\ln f)$. $\ln d_M$ varies from $0$ to $\frac12 \log m$. Curve I is $R=\ln d_M$. Curve II is $R=\ln d_M+\tau(f\ln f)$. Curve III is $R=\tau(f\ln f)$. Curve IV is  $R=\frac12(\ln m+\tau(f\ln f))$. The real values of $Q^{(1)}$ and $Q$ are in the quadrilateral surrounded by four lines. When $\ln d_M$ is small, our estimates are tight. This is the figure for $\ln d_M=\frac14 \ln m$.}
\end{figure}

%% file: Qcap_examples.tex
\section{Examples}
\subsection{Group channels}
Starting from a finite group $G$, we will construct two classes of channels. We will use the
quantum group framework \cite{JNR} for both of these constructions. From a harmonic analysis point of view, group channels were also discussed in \cite{Neufang} for general locally compact groups. We restrict ourselves to finite groups here.

\subsubsection{Hadamard channels}
Generalized dephasing channels, as a special case of Hadamard channels, are called Schur multipliers in the operator algebra literature. The Hadamard channels are known to be degradable (see \cite{DS}), hence the quantum capacity does not require regularization, i.e. $Q^{(1)}=Q$. Our estimates overlap with  the quantum capacity formula in \cite{Neufang} for finite groups, but both approaches are based on the unfortunately unpublished joint work \cite{JNR2}. The arguments, however, are different. Our approach  provides a new proof of
$Q=Q^{(p)}$ for these particular Schur multipliers, but this is already known thanks to the fact that Hadamard channels are strongly additive for $Q^{(1)}$ \cite{WD}.

Suppose $G$ is a finite group with order $|G|=m$ and  $1$ as its identity. We denote the group von Neumann algebra by $L(G)$, the algebra generated by $\{\la(g)|g\in G\}$ . Here $\lambda(g)$ is the left shift unitary defined on $B(l_2(G))$ as follows
\begin{align*}
\la(g)(e_h)=e _{gh}\pl, \ \ \ \ \ \forall h\in G \pl ,\end{align*}
where $\{e_h|h\in G\}$ is the canonical basis of $l_2(G)$, i.e.  $e_h(g)=\delta_{h,g}$. The algebra
of functions $l_\infty(G)$ is dual to $L(G)$ in sense of quantum groups and sits as diagonal matrices in $B(\ell_2(G))$. Let us denote by $e_{g,g}$ the diagonal matrix unit. Then $f=\sum_g f(g)e_{g,g}$ is in $l_\infty(G)$. The normalized traces on $L(G)$ and $l_\infty(G)$ are $\tau$ and $\tau'$ respectively
\begin{align*}\tau(\sum_g \al(g) \la(g))=\al(1),\ \ \tau'(\sum_g f(g)e_{g,g})=\frac1m\sum_g f(g) \pl .\end{align*}
We note that $L_2(l_\infty(G),\tau')\cong L_2(L(G),\tau)\cong l_2(G)$, and $l_\infty(G)\subset B(l_2(G))$, $L(G)\subset B(l_2(G))$ are both standard inclusions. The matrix Schur multiplication (or Hadamard product) is given by (here and in this section ``$*$'' always denotes the Schur multiplication for two matrices)
\begin{align*}(a_{ij})*(b_{ij})=(a_{ij}\cdot b_{ij}).\end{align*}
It is a well-known fact (see \cite{Pau}) that the multiplier map for a given matrix $a=(a_{ij})$,
\begin{align*}M_a (b)=a* b\ \ \ \text{for}\  b=(b_{ij})\in M_m\pl ,\end{align*}
is completely positive if and only if $a$ is positive. Moreover, $M_a$ is trace preserving if and only if $a_{ii}=1$ for $1\le i\le m$. In our situation, we further restrict the matrix $a$ to be a density in $L(G)$. The Stinespring unitary has the following form
\begin{align*}U=\sum_g e_{g,g}\otimes \la(g) \in l_{\infty}(G)\ten L(G)\pl .\end{align*}
This means $N=L(G)$ will be considered as the algebra of symbols, and $M=M'=l_{\infty}(G)$. The VN-channel depending on a density $\rho=\sum \rho(g)\la(g)\in L(G)$ is defined as follows,
\begin{align*}\theta_\rho(\omega)&=id \otimes \tau[(\sum_g e_{g,g}\otimes \la(g)\omega\ten \rho  (\sum_g e_{g,g}\otimes \la(g)^*)]\\&=\sum_{g,g'} \tau (\la(g)\rho \la(g')^*)e_{g,g} \omega e_{g',g'} \lel (\rho(g^{-1}g')){\rm *}(\omega_{g,g'})\ \pl ,
\end{align*}
where $\omega=\sum_{g,g'} \omega_{g,g'}e_{g,g'}\in S_1(l_2(G))$. This is a Schur multiplier by a density in $R(G)$. It is obvious that $\ket{e_{g,g}}$ and $\ket{\la(g)}$ are two orthogonal bases in $L_2(M)$ and $L_2(N)$ respectively. Hence Theorem \ref{best} applies, we obtain
\begin{enumerate}
\item[i)] $-S_{cb}(\theta_\rho)=Q^{(1)}(\theta_\rho)=Q(\theta_\rho)=Q^{(p)}(\theta_\rho)=\tau(\rho\ln \rho)\pl ,$
\item[ii)] Since $\theta_\rho$ is unital, we have \[-S_{cb}(\theta_\rho)+\ln m=C_{EA}(\theta_\rho)=2Q_{EA}(\theta_\rho)= \ln m+\tau(\rho\ln \rho)\pl,\] and these are attained at a maximally entangled state.
\end{enumerate}
 Note here $M=l_{\infty}(G)$ is commutitave, we have $\ln d_{l_\infty(G)}=0$. Thus in Figure.1 the Curve II and Curve III coincide and give the equality. In \cite{Neufang}, the formula for $Q(\theta_\rho)$ is obtained differently.
\begin{exam}\label{dephasing}{\rm
A well-studied qubit example is the dephasing channel. Let $0\le q\le 1$ be the dephasing parameter, we have
\begin{align*}
\Phi_q \large(\big[ \begin{array}{cc}
    a&b\\
    c&d
  \end{array}\big]\large)=\big[\begin{array}{cc}
    a&qb\\
   q c&d
  \end{array}\big]\pl .
\end{align*}
The channel  can also be expressed using the Pauli matrix $Z=\left[
\begin{array}{ccc}
    1& 0\\
  0&  -1\\
  \end{array}\right]$,
\begin{align*}
\Phi_q(\rho)=(1-\fs{1-q}{2})\rho+\fs{1-q}{2}Z\rho Z\pl .
\end{align*}
This corresponds to $G=\mathbb{Z}_2$ for $\rho=1+qX=\big[\begin{array}{cc}
    1&q\\
   q &1
  \end{array}\big]$ in our setting. We obtain $Q(\theta_\rho)=\tau(\rho\ln \rho)=\ln 2-H(\fs{1+q}{2})$, which is same with the formula in \cite{HW}.}
\end{exam}
When the dimension $m>2$, we cannot recover an arbitrary generalized dephasing channels via the group construction, because the class of channels $\theta_{\rho}$ is a strict subset of all Schur multipliers.
\subsubsection{Random unitary}
 A channel map is called a random unitary channel if it is a convex combination of unitary conjugation. Again, we use the shift unitaries $\{\la(g)\}$ defined above and $U=\sum_g e_{g,g}\otimes \la(g)$ as the Stinespring unitary defined as in the previous case.
 We switch, however,  the roles of the environment and output. This means we consider $M'=L(G)$ and the symbol algebra $N=l_{\infty}(G)$. Thus $M=R(G)$ as the right group von Neumann algebra generating by right shift unitary $\{r(g)|g\in G\}$. For each density $f\in l_\infty(G)$, we define the VN-channel by
\begin{align*} \theta_f(\rho)&=\tau'\ten id (U(f\ten \rho)U^*)=\sum_{g,g'} \tau'(e_{g,g} f e_{g',g'}) \la(g)\rho \la(g')^*\\&=\fs{1}{m}\sum_g f(g)\la(g)\rho \la(g)^*\ , \ \  \  \forall\  \rho\in S_1(l_2(G))\pl .\end{align*}
Two extreme cases are $f=m \pl e_{g,g}$ and $f=1$. The former one is a perfect unitary conjugation channel by $\la(g)$, and the latter one is the conditional expectation onto ${M=R(G)}$. Thanks to the Peter-Weyl theorem, here the index  $d_{R(G)}$ is the largest degree of irreducible representations, or the dimension of the largest irreducible representations. For short, we denote $d_G\equiv d_{R(G)}$. Theorem \ref{best} implies,
\begin{enumerate}
\item[i)] $\max \{\ln d_{G}, \tau(f\ln f)\}\le \! Q^{(1)}(\theta_f)\!\le\! Q(\theta_f)\le Q^{(p)}(\theta_f)\!\le \!\tau(f\ln f)+\ln d_{G} ;$
\item[ii)] $-S_{cb}(\theta_f)+\ln m=C_{EA}(\theta_f)=2Q_{EA}(\theta_f)= \ln m+\tau(f\ln f)$ is attained at a maximally entangled state.
\end{enumerate}
\re When the group $G$ is abelian, $R(G)$ is a commutative algebra. Then $d_G=0$, so upper and lower bounds coincide as the Hadamard channels:
\begin{align*}-S_{cb}(\theta_f)=Q^{(1)}(\theta_f)=Q(\theta_f)=Q^{(p)}(\theta_f)= \tau(f\ln f)\pl .\end{align*}
In this case, we have $R(G)\cong l_\infty (\hat{G})$ with $\hat{G}$ being $G$'s dual group. For finite $G$, $G\cong \hat{G}$ so $\theta_f$ are also Hadamard channels.
\mar
\begin{exam}\label{dihe}{\rm
  The qubit example is the bit-flip channel. Let $G=\mathbb{Z}_2$, the nontrivial shift unitary is the pauli matrix $X=\left[
\begin{array}{ccc}
    0& 1\\
  1&  0\\
  \end{array}
\right]$ . For the flip parameter $0\le q\le 1$ and qubit density $\rho\in S_1^2$,
\[
\Phi_q(\rho)=(1-q)\pl\rho+q\pl X\rho X\pl .
\]
One can see this is unitarily equivalent to the dephasing channel in Example \ref{dephasing} with dephasing parameter $\frac{1-q}{2}$.}
\end{exam}

In general the degree of the largest irreducible representation is not $1$, unless $G$ is commutative. There are several facts in representation theory giving upper bounds for the integer $d_G$. One we will use below is that if $H\subset G$ as an abelian subgroup, then $\max_k n_k\le [G:H]$. We will compare the two upper bounds for $Q$ in the following examples.
\begin{exam}{\rm
For the dihedral groups $D_{2n}$, the group of symmetries of a $n$-regular polygon \cite{dummit}, our estimates are almost optimal. Indeed, for dihedral groups $d_{D_{2n}}$ is always $2$ for any $ \nen$. So our estimates control everything up to one qubit \begin{align*}\max \{\ln 2, \tau(f\ln f)\}\le Q^{(1)}(\theta_f)\le Q(\theta_f)\le Q^{(p)}(\theta_f)\le \tau(f\ln f)+\ln 2\pl .\end{align*}
When $n$ is large and $f$ is close to pure states, $\ln 2$ is small compared to $\tau(f\ln f)$ .}
\end{exam}

\begin{exam}{\rm Let $G$ be the semi-product group $\Z_d^l\rtimes \Z_l$, where $\Z_d^l$ is the $l$ direct sum of cyclic groups $\Z_d$, $\Z_l$ does the shift action as follows,
\begin{align*}
(x_1,x_2,\cdots, x_d, j)(x_1',x_2',\cdots, x_l', j')=(&x_1+x_{1+j}',x_2+x_{2+j}',\cdots, x_l+x_{l+j}', j+j')\pl,
\end{align*}
for any $1\le x_i,x'_i\le d,\ 0\le i, j\le l\pl.$
%Here we consider the addition $i+j$ in $\Z_l$.
Note that since $\Z_d^l$ is an abelian subgroup of $G$, then it is easy to see that  $d_{G}\le l$. The comparison bound is better when $\tau(f\ln f)\le l\ln d-2\ln l$. When $d$ is large, $\ln l\ll l\ln \sqrt{d}= |G|^{\frac12}$.}
\end{exam}
\begin{exam}\label{asym}{\rm For the symmetry group $|S_n|=n!$, it is shown in \cite{asym} that there exists constants $c_1, c_2>0$ such that
\begin{align*} -c_1\sqrt{n}\le d_{S_n}-\frac{1}{2} \ln {n!}\le -c_2\sqrt{n} \pl .\end{align*}
This implies that the comparison bound is better if $\tau(f\ln f) \le 2c_2\sqrt{n}$ and the upper bound via $Q_{EA}$ bound is better when $\tau(f\ln f) \ge 2c_1\sqrt{n}$. Note although $[0,2c_2\sqrt{n}]$ is a relatively small region in the range of $\tau(f\ln f)$ (since $n \ll n! =|G|$), it is a definitely gaining part of the comparison estimate when the density $f$ is slightly perturbed from the identity $1$.}
\end{exam}
\subsection{Pauli channels}
Pauli channels are by no means optimal for the comparison bounds, but they do fit in our framework. Pauli channels
are convex combinations of unitary conjugations by Pauli matrices. In high dimensions, we may interpret the Heisenberg-Weyl operators as the generalized Pauli matrices \cite{Wildebook}. These operators are used to establish teleportation and superdense coding in high dimension. Let us consider $\{e_k|1\le k\le n\}$ as the standard basis of an $n$-dimensional complex Hilbert $H= l_2^n$. The generalized Pauli matrices $X$ and $Z$ for an $n$-dimensional system are
\begin{align*}
X(e_k)=e_{k+1},\ \  Z(e_k)=\exp ({\fs{2k\pi i}{n}})e_k\ \  \text{for} \ \ 1\le k \le n\pl .
\end{align*}
For $k=n$ we use the convention $e_{n+1}=e_1$. $X$ and $Z$ satisfy the commutation relations,
\begin{align*}
XZ=\exp({\fs{2k\pi i}{n}})ZX\pl .
\end{align*}
Now an $n$-dimensional Pauli channel can be defined as follows,
\begin{align*}
\theta_f(\rho)=\fs{1}{n^2}\sum_{1\le i, j \le n}f_{ij} X^iZ^j \rho (X^iZ^j)^*\pl .
\end{align*}
In order to be a channel, the coefficient $f_{ij}$ must satisfy $f_{ij}\ge 0, \sum f_{ij}=n^2$. Now we consider $f\in N=l_{\infty}^{n^2}\subset B(l_{2}^{n^2})$, where $N$ is the commutative algebra spanned by ${\{P_{ij}|1\le i,j\le n\}}$ as its rank one projections. The normalized trace (which makes the operator $B$ a unitary) is given by  $\tau(f)=\fs{1}{n^2}\sum_{ij} f(ij)$. We have the Stinespring dilation,
\begin{align*}
\theta_f(\rho)=\sum_{1\le i, j \le n} \tau (P_{ij}f)X^iZ^j \rho (X^iZ^j)^* = id\otimes \tau (U (\rho \otimes f) U^{*}) \pl ,
\end{align*}
where $U$ is a joint unitary in $B(l_2^n)\otimes N$,
\begin{align*}
U=\sum_{1\le i,j\le n} X^iZ^j\otimes P_{ij}\pl .
\end{align*}
One can easily see that $\theta_f$ is unital and $U$ satisfies the assumptions of Theorem \ref{cbentropy}. Indeed $\{X^iZ^j|1\le i, j\le n\}$ is an orthogonal basis for $M_n$ and $\{P_{ij}| 1\le i, j\le n\}$ is an orthogonal basis for $L_2(\ell_{\infty}^{n^2}))$. Thus by Corollary \ref{lowerbound} we deduce that \begin{align}\label{scb}-S_{cb}(\theta_f)=\tau(f\ln f)-\ln n,\ C_{EA}(\theta_f)=2Q_{EA}(\theta_f)= \tau(f\ln f)\pl .\end{align}
For the comparison bound, we consider $\theta_f\ten id_{M_n}$ instead of $\theta_f$. Note that ${\{X^iZ^j\ten 1}_{i,j}$ is an orthogonal  basis for $M_n\otimes 1$ and $M_n\ten 1\subset M_n\ten M_n$ is a standard inclusion as in the Example \ref{GNS}. This allows us to apply Theorem \ref{comp} and its corollary:
\begin{align*}\ln n\le Q^{(1)}(\theta_f\ten id_n)\le Q(\theta_f\ten id_n)\le Q^{(p)}(\theta_f\ten id_n)\le \tau(f\ln f)+\ln n \pl. \end{align*}
Note that $Q^{(p)}$ is subadditive, we find
 \begin{align*}
   Q^{(p)}(\theta_f\ten id_n)= Q^{(p)} (\theta_f)+\ln n \pl .
   \end{align*}
Hence
 \begin{align*}0\le Q^{(1)}(\theta_f)\le Q(\theta_f)\le Q^{(p)}(\theta_f)\le \tau(f\ln f)\pl .\end{align*}
Thus for generalized Pauli channels, the comparison bound is always outperformed by \eqref{scb} and entanglement assistance, i.e. $Q(\theta_f)\le Q_{EA}(\theta_f)=\fs{1}{2}\tau(f\ln f)= \ln n -\frac{1}{2}H(\fs{1}{n^2}f)$ (because $\dim N =n^2$). This in the Figure.1 corresponds to the case $\ln d_M= \frac{1}{2}\ln m$, and hence the Curve IV is always lower then the Curve II. However, by applying an averaging trick, we obtain an new bound for potential quantum capacity $Q^{(p)}$ for high dimension depolarizing channel.
\begin{exam}{\rm The $d$-dimensional depolarizing channel with parameter $q\in [0,1]$ is
 \[ \D_{q}(\rho)\lel q \rho+(1-q)\frac{1}{d} \pl.\]
The depolarizing part $\rho\to \frac{1}{d}$ is actually the generalized Pauli channel with uniform distribution,
\[\frac{1}{d^2}\sum_{i,j}X^iZ^j\rho (X^iZ^j)^*=tr(\rho)\frac{1}{d}.\]
Then $\D_{q}$ is the Pauli channel with the distribution $f_{00}=q+\frac{1-q}{d^2}$, $f_{ij}=\frac{1-q}{d^2}$ for $(i,j)\neq (0,0)$.
Let us first consider the following dephasing channel
\[\Phi_{q'}(\rho)= q' \rho+(1-q')\E(\rho) \pl,\]
where $\E$ is the conditional expectation onto the diagonal matrices (the completely dephasing channel) and $q'\in [0,1]$. This channel dephases the off diagonal entry by a factor $q'$ and by the discussion of 8.1.1 we know
\[Q^{(p)}(\Phi_q')=\log d- \frac{(d-1)q'+1}{d}\log \frac{(d-1)q'+1}{d}-\frac{(d-1)(1-q')}{d}\log\frac{(1-q')}{d}\pl.\]
Similarly, the channel $\rho\to U^*\Phi_q'(U\rho U^*)U$ is also a dephasing channel but to the basis $\{Ue_i\}_i$ instead of the standard basis $\{e_i\}_i$. We claim that the averaging of dephasing channels uniformly on all basis will give us a depolaring channel. Namely for any state $\rho\in M_d$
\[\int_{U(M_d)} U^*\E(U \rho U^*)U=\frac{1}{d+1}\rho +\frac{1}{d+1}\pl. \]
This can be proved by the averaging the Choi matrix. Denote $\E_U=U^*\E(U \cdot U^*)U$, let $\psi_d$ be the maximally entangled state $\sum_{i=1}^d e_i\ten e_i$, then
\begin{align*}\chi_{\E_U}&=id\ten\E_U(d\ket{\psi_d}\bra{\psi_d})=d\pl id\ten U^*\E(1\ten U\ket{\psi_d}\bra{\psi_d}1\ten U^*)U\\&= d id\ten U^*\E(U^t\ten 1\ket{\psi_d}\bra{\psi_d}\bar{U}\ten 1)U\\&=(U^t\ten U^*)id\ten \E(d\ket{\psi_d}\bra{\psi_d})(\bar{U}\ten U)\\&=(U^t\ten U^*)\chi_{\E}(\bar{U}\ten U)\end{align*}
Note that $\chi_\E=\sum_{i=1}^d e_{i,i}\ten e_{i,i}$ and hence the partial transpose on first component gives us
\[t\ten 1((U^t\ten U^*)\chi_{\E}(\bar{U}\ten U))=(U^*\ten U^* )\chi_\E(U\ten U)\pl.\]
By representation theory (\cite{rep}, Proposition $2.2$), we have
\[\int_{U(M)}\chi_{\E_U}=\frac{d}{d+1}\ket{\psi_d}\bra{\psi_d}+\frac{1}{d+1}1\ten 1\pl,\]
which proves the claim. Then for averaging the $q'$-dephasing channel, we have
\[\int_{U(M_d)} U^*\Phi_{q'}(U\rho U^*)U=(q'+\frac{1-q'}{d+1})\rho + \frac{1-q'}{d+1}=D_{q'+\frac{1-q'}{d+1}}(\rho)\pl.\]
Set $q'+\frac{1-q'}{d+1}=q$, by convexity of $Q^{(p)}$ we get
\begin{align} \label{new}
   Q^{(p)}(\D_q) \le &\log d -H(\frac{q(d^2-1)+1}{d^2})\\&-\frac{(d^2-1)(1-q)}{d^2}\log (d-1)
    \pl . \nonumber
\end{align}
It is known that for $q\!\!\le \!\!\frac{1}{d+1}$ the channel $\D_{p}$ becomes entanglement-breaking (it is an averaging of completely dephasing channel.) and hence $Q^{(p)}(\D_{\frac{1}{d+1}})=0$ (see \cite{WD}).
This upper bound \eqref{new} vanishes at $p= 1/(d+1)$ and is convex in the interval $[1/(d+1),1]$. For $d=2$ it is \cite{Winterss} proved the upper bound
\begin{align}
   Q^{(p)}(\D_p) \le  1 -H(\frac{3p+1}{4})\pl,\nonumber
   \end{align}
by using a convex combination of dephasing channels to Pauli-$X,Y,Z$ basis.
Using the unitaries from teleportation one can generalize their method to higher dimension, but that upper estimate only yields the first two terms in \eqref{new}. Since the third term is negative, our upper bound are tighter for $d>2$.
\begin{figure}
\centering
\includegraphics[width=0.5\textwidth]{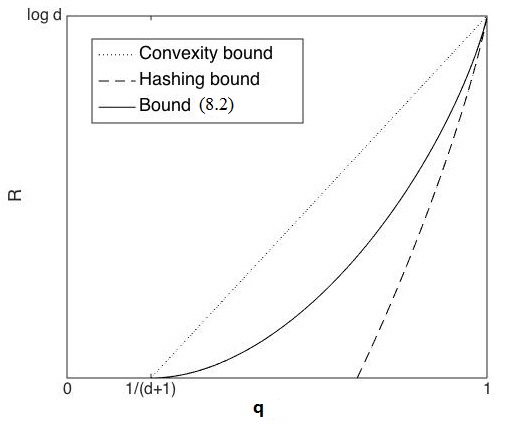}
\caption{Our upper bound for the ``one-shot'' potential quantum capacity of a $d$-dimenional depolarizing channel: The dotted line is the convexity bound by the fact, $Q^{(p)}(\D_{\frac{1}{d+1}})=0$. The dashed curve is the hashing (lower) bound. The solid curve is our new upper bound \eqref{new}. This is the figure for $d=5$.}
\end{figure}}
\end{exam}
\subsection{Majorana-Cliffords}
The fourth class example we consider is Clifford algebra. The Clifford algebra $\mathit{Cl}_n$ has $n$ generators $\{C_i\}_{1\le i\le n}$, which satisfy the CAR (canonical anti-commutative relations):
\[
C_i= C_i^*\ \ , \  C_iC_j+C_jC_i=2 \delta _{ij}\ \ \ \text{for}\ \  \forall 1\le i, j\le n \pl .
\]
The self-adjoint property $C_i = C_i^*$ has a physical interpretation as creation and annihilation operators for Majorana fermions. Proposed candidates for Majorana fermions include supersymmetric analogs of bosons, dark matter, neutrinos, and electron-hole superpositions in topological condensed matter systems \cite{wilczek,kitaev-unpaired}. Recent experiments have observed evidence of Majorana fermions in such condensed matter systems \cite{nadj-perge,li,bunkov}. Condensed matter Majorana modes may serve as the basis for topological quantum computers \cite{wilczek}, such as the physical motivation for the Drinfeld Double example below.

It is a known fact that $\mathit{Cl}_n$ is isomorphic to $2^n$-dimensional matrix algebra $M_{2^n}$. We have the canonical orthogonal basis of $L_2(\mathit{Cl}_n, tr)$ defined by $\{C_A| A \subset [n]\}$, where $[n]=\{1,2,3\cdots ,n\}$ and
\[
C_A=\Pi_{i\in A}C_i:=C_{i_1}C_{i_2}\cdots C_{i_k} \pl \text{with}  \pl i_1<i_2<\cdots <i_k\ \  \text{and}\ \  \{i_1, i_2, \cdots, i_k\}=A\subset [n]\pl .
\]
The order of the product matters because of the CAR. Similar to Pauli channels, let us set $N=l_{\infty}(2^{n})$ equipped with normalized trace $\tau$. The Stinespring unitary is
\[
U=\sum_{A\subset [n]} C_A\otimes P_A \in M_{2^n}\otimes N\pl .
\]
For a density (probability distribution) $f\in l_{\infty}(2^{n})$, we can define a Clifford channel
\[
\theta_f(\rho)\lel id\otimes \tau (U(\rho\otimes f) U^*)\lel \frac{1}{n^2}\sum_{A\subset [n]}f(A)  C_A\rho{C_A}^{*} \pl ,
\]
as random unitaries. By Theorem \ref{best}, we obtain similar results as Pauli channels,
\begin{align}-S_{cb}(\theta_f)=\tau(f\ln f)-n\ln 2\le Q^{(1)}(\theta_f), \ C_{EA}(\theta_f)=2Q_{EA}(\theta_f)= \tau(f\ln f)\pl .\end{align}
Again the upper bound via $Q_{EA}$ is tighter than the one given by the comparison theorem.

\subsection{Quantum group channels}
A finite dimensional quantum group is a Hopf algebra with an antipode. Quantum groups form a class of Hopf algebras that contains groups and their duals. More precisely, we are given a (finite dimensional) algebra $A$ and a $*$-homomorphism $\Delta:A\to A\ten A $ and the co-multiplication which satisfies
 \[ (\Delta\ten id_A)\Delta \lel (id_A\ten \Delta)\Delta \pl .\]
For locally compact quantum groups the antipode is determined by the left and right Haar weight (see \cite{vaes}). Finite dimensional quantum groups are of Kac-type. For us this means that we have a trace $\tau$ such that
 \[ (\tau\ten id)\Delta(x) \lel \tau(x)1 \lel (id\ten \tau)\Delta(x) \pl .\]
More importantly every quantum group (of Kac-type, see \cite{baaj,kac}) admits a (multiplicative) unitary $V\in B(L_2(A))\ten B(L_2(A))$ such that
 \[ \Delta(x) \lel V(x\ten 1)V^*  \pl .\]
Moreover, $V\in \hat{A}\ten A$ (see \cite{baaj} Section 3.6 and 3.8) with dual object $\hat{A}$. Following \cite{JNR} we may define
 \[ \theta_f^{\dag}(T) \lel id\ten \tau((1\ten f)V(T\ten 1)V^*)\pl, \pl \theta_f(\rho)\lel id\ten \tau(V^*(\rho\ten f)V) \pl .\]
Here $\theta_f^{\dag}$ is the adjoint map of the channel $\theta_f$.
Thus we find the Stinespring unitary $U=V^*\in \hat{A}\ten A$ and $\Theta:L_1(A,\tau)\to CB(S_1(L_2(\hat{A}'))$ the channel map. Here we may and will assume that $\tau$ is the restriction of the normalized trace on $B(L_2(A))$.  Thus we set $N=A$ and $M=\hat{A}'$, and they are of the same dimension. It was shown in the unpublished paper \cite{JNR2} that $B$ corresponds to the Fourier transform, and hence sends an orthonormal basis in $L_2(\hat{A}')=L_2(\hat{A},\tau)$ to an orthonormal basis in $L_2(A,\tau)$. Therefore the assumptions of Theorem \ref{best} are all satisfied and in particular,
 \[Q^{(1)}(\theta_f)\le  Q(\theta_f)\le Q^{(p)}(\theta_f)\le 2 Q^{(1)}(\theta_f) \pl .\]

\begin{rem} {\rm Here we have an trivial but interesting observation. Let $A$ be a finite dimensional quantum group with representation ${A=\oplus_k M_{n_k}}$. It is easy to see from representation theory $d_A\le \sqrt{\dim A}$. On the other hand, we perform the construction above for $\hat{A}'$ instead of $A$. Then $\theta_1$ is the conditional on $A$, and hence
 \[ \ln d_{A} \lel Q(\theta_1)\le \frac{1}{2}C_{EA}(\theta_1) \lel \frac{\ln n+\tau(1\ln 1)}{2}= \frac{\ln n}{2}\]
where $n=\dim(A)=\dim(A')$. This gives quantum information perspective of $d_A\le \sqrt{\dim A}$. }
\end{rem}
\subsection{Crossed product}
 %We consider the Drinfeld double that can be represented as the crossed product ${\ell_{\infty}(G)\rtimes_{\al} G}$ with respect to conjugation action $\al$ (defined below). This example is also a quantum group, in fact one of the first examples of neither commutative nor co-commutative quantum groups (see \cite{drinfeld}).
Our particular Hadamard channels in 8.1.1 and random unitaries in 8.1.2 are quantum group channels for commutative or co-commutative symbol algebra.
 Here we will use crossed products to build a mixture of these two. A connection is found in Kitaev's work on quantum computation by anyons \cite{kitaev}. Given a finite group $G$, we consider the operators $\{A_g, B_g\pl |\pl  g\in G\}$ satisfying the following relations
\begin{align}A_h A_g = A_{hg} \ ,\  B_g B_h = \delta_{g,h} B_g \ , \  A_g B_h = B_{ghg^{-1}} A_g\pl,\pl \forall g,h\in G \pl.\label{cmr}\end{align}
They are the local gauge transformations and magnetic charge operators for vertices on a two-dimesional lattice in which edges correspond to spins. The crossed product corresponds to an algebra of local operators, which commute with the topological operators used to perform quantum computations. For this reason, the local operators generating the crossed product leave a significant subspace invariant, which in Kitaev's physics corresponds to the space of degenerate ground states. This means that the anyonic quantum computer is naturally immune to local perturbations, possibly obviating the need for active error correction and presenting a quantum computation paradigm that resists decoherence due to its underlying physical structure.

 Now consider $l_\infty(G)\subset B(l_2(G))$ as the diagonal matrices. Define the action $\al$ of $G$ acting on $l_\infty(G)$ as automorphism \[\al_g(e_{h,h})=W_ge_{h,h}W_g^*=e_{ghg^{-1},ghg^{-1}}\pl,\]
  where $W_g(e_h)=e_{ghg^{-1}}$ are unitary in $B(l_2(G))$. The (reduced) crossed product $M=l_\infty(G)\rtimes_\alpha G $ is defined to be the algebra generated by the range of the following two representations on $l_2(G,l_2(G))\cong l_2(G)\ten l_2(G)$,
 \begin{align*}
\pi: &l_\infty(G)\rra B(l_2(G)\ten l_2(G))\pl ,\ \ \ \ \pi(x)=1\otimes x\ ;\\
\tilde{\lambda}: & G\ \rra B(l_2(G)\ten l_2(G))\pl ,\ \ \ \
\tilde{\lambda}(g)=\lambda(g)\otimes W_g \pl ,
\end{align*}
where $\lambda$ is the left regular representation of group $G$. We observe that $M,M'\subset B(l_2(G\times G))$ is a standard inclusion, and the operator $J$ and commutant $M'$ are given as follows,
\begin{align*}J(e_g\otimes e_h)=e_{g^{-1}}\otimes e_{g^{-1}hg}\pl ,\pl
J\pi(x)J=\sum_g e_{g,g}\otimes W_gxW_g^*\ ,\ \ J\tilde{\lambda}(g)J=r(g)\otimes 1\pl.\end{align*}
Thus neither $M$ nor $M'$ is commutative. Denote $A_g=\la(g)\otimes W_g$ and $ B_h=1\otimes e_{h,h}$, one can check they satisfy the commutation relations \eqref{cmr} in Kitaev's setting. Now we are ready to use these operators to construct channels.\\
Case 1. Consider the Stinespring unitary $U\in M'\ten B(l_2(G\times G))$
\[U=\sum_{g,h} (A_gB_h)\otimes (\lambda(h)\otimes e_{g,g})\pl ,\]
with the first bracket elements in $M$ and second bracket in $N=L(G)\bar\otimes l_\infty(G)$. For $f\in L(G)\otimes l_\infty(G)$, we can write $f=\sum_g f_g \otimes e_{g,g}$, where each $f_g=\sum_h f_g(h)\la (h)\in L(G)$. The channel for a density $f\in N$ is defined as follows,
\begin{align*}
&\theta_f: S_1(l_2(G\times G))\rra S_1(l_2(G\times G))\\
& \theta_f(\rho)=\sum_{g,g',h, h'} \tau[(\lambda(h)\otimes e_{g,g})f(\lambda(h')^*\otimes e_{g',g'})]A_gB_h\rho (A_{g'}B_{h'})^*\\
&=\sum_{g,h, h'} f_g(h'^{-1}h)\la(g)\rho_{h,h'}\la(g)^*\otimes e_{ghg^{-1},gh'g^{-1}},
\ \ \forall \pl\rho=\sum_{g,h\in G} \rho_{h,h'}\otimes e_{h,h'}\in S_1(l_2(G\times G))\pl .
\end{align*}
One can see that this channel is a mixture of random unitary and Schur multiplier. It is unital because
\begin{align*}\theta_f(1)=\sum_{g,h}\tau(f_g) 1_{B(l_2(G))}\ten  e_{ghg^{-1},ghg^{-1}}=1_{B(l_2(G)\ten l_2(G) )} \pl.\end{align*}
It is easy to check that $U$ satisfies assumptions of Theorem \ref{best}. Note that $dim M=n^2$, we have
\begin{enumerate}
\item[i)] $-S_{cb}(\theta_f)=\tau(f\ln f), \ C_{EA}(\theta_f)=2Q_{EA}(\theta_f)= \tau(f\ln f)+2\ln n$;
\item[ii)] $\max\{d_M, \tau(f\ln f)\}\le Q^{(1)}(\theta_f)\le Q(\theta_f)\le Q^{(p)}(\theta_f)=\tau(f\ln f)+\max_k \ln n_k\pl .$
\end{enumerate}
Case 2. Consider another unitary
\[U'=\sum_{g,h} (A_gB_h)\otimes e_{hg,g} \pl .\]
Now the symbol algebra $N$ is $B(l_2(G))$. For a density, $f=\sum_{g,h}f_{g,g'}e_{g,g'}\in N$, we define the channel $\theta_f: S_1(l_2(G\times G))\to S_1(l_2(G\times G))$ associated with $f$ as
\begin{align*}
\theta_f(\rho)=\sum_{g,g',h, h} \tau(e_{hg,g}fe_{g',h'g'})A_g B_h\rho  (A_gB_h)^*=\fs{1}{n}\sum_{h g=h' g'} f_{g,g'}(\la(g)\rho_{h,h'}\la(g)^*\otimes W_ge_{h,h'}W_{g'}^*),
\end{align*}
for any $ \rho=\sum_{h,h'\in G} \rho_{h,h'}\otimes e_{h,h'}\in S_1(l_2(G\times G))\pl .$
Again it is unital, so our theorem give the same estimates as case 1.
\subsection{Non-unital channels}
So far the examples above are unital channels. In this part, we provide a non-unital example for which our estimates still apply. Let $G$ be a finite group of order $m$, and $g,h\in G$ be its group elements. Denote  $B(l_2(G))\cong M_m$ and $e_{g,h}$ as the matrix units. Consider the Stinespring unitary
\[U=\sum_{g,h\in G} e_{gh,h}\otimes e_{g,gh}\in M_m\ten M_m\pl .\]
For each density $f\in (M_m,\fs{1}{m}tr)$ (for the symbol algebra we use the normalized trace), we may define $\theta_f: S_1^m\rra S_1^m$ as follows
\begin{align*}
\theta_f(\rho)&=\frac1m\sum_{g,h,h'\in G} f_{gh,gh'}\rho_{h,h'}e_{gh,gh'}=\frac1m\sum_g f*(\lambda(g)\rho\lambda(g)^*)\\
&= f*(\frac1m\sum_g\lambda(g)\rho\lambda(g)^*)\pl ,\ \ \ \ \
 \forall \rho=\sum_{h,h'}\rho_{h,h'} e_{h,h'}\in S_1(l_2(G))\pl .
\end{align*}
Here ``$*$'' is again the Schur multiplication and $\lambda$ is the left regular representation. One can see that this channel is a composition of a random unitary and a Schur multiplier. In general this channel is not unital,
\[
\theta_f(1)=\frac1m\sum_g f* 1=\E(f)\pl .
\]
Here $\E$ denote the conditional expectation onto the diagonal matrices $f=\sum_g f_{g,g} e_{g,g}$. Since $\{e_{gh,h}\}$ and $\{e_{g,gh}\}$ are orthogonal basis of the full matrix algebra $M_m$, Theorem \ref{best} implies
\begin{align}\label{nu}&-S_{cb}(\theta_f)=\tau(f\ln f)-\ln m,\ \  \ \   C_{EA}(\theta_f)=2Q_{EA}(\theta_f)\le \tau(f\ln f)\pl, \nonumber\\
&H(\frac1m \E(f))-H(\frac1m f)\le Q^{(1)}(\theta_f)\le Q(\theta_f)\le \frac12 \tau(f\ln f) \pl . \end{align}
In particularly, we know $H(\frac1m \E(f))-H(\frac1m f)\ge 0$, because unital channels always increase the entropy. As for Pauli channels, the comparison estimates apply for $id\ten \theta_f$ instead of $\theta_f$, but \eqref{nu} is tighter than the comparison estimates.\\

%\section{Conclusion}\noindent
%\indent We identify ``nice'' classes of channels which admit quantum capacity estimates from above and below. These nice classes of channels share similar Stinespring dilations. Our result rest on particular algebraic property of their joint unitary. It gives insights on the gap between $Q^{(1)}$ and $Q$ or even $Q^{(p)}$ and also consider the capacity change on channels depending a quantum parameter. One can understand the estimates as a perturbation theory to the conditional expectations. We generalize a cb-entropy formula in \cite{JNR2} and reprove, from a different approach, a quantum capacity formula in \cite{Neufang}. The present work introduces the interpolation technique to obtain entropy and capacity inequality and also shows more connections between capacities and operator space norms, as has been demonstrated in \cite{DJKR,GW,JP,psumming}.

%One natural speculation is that whether this method can be generalized to more general classes of channels. Also, an information theoretic description to the ``nice'' channels or a quantum information way to prove the entropy inequalities (Corollary \ref{entropyin}) will be interesting. This might shed some operational meaning on the operator space techniques we used. 

%% file: Qcap_bib.tex
\noindent\emph{Acknowledgement}---We thank Mark M. Wilde for helpful discussion and passing along the reference \cite{cubitt}, Andreas Winter and Debbie Leung for interesting remarks on the potential quantum capacity, and Carlos Palazuelos for continuing discussions on capacities. MJ is partially supported by NSF-DMS 1501103. NL is supported by NSF Graduate Research Fellowship Program DGE-1144245.

%% file: Qcap.bbl
\begin{thebibliography}{10}
\providecommand{\url}[1]{\texttt{#1}}
\providecommand{\urlprefix}{URL }
\expandafter\ifx\csname urlstyle\endcsname\relax
  \providecommand{\doi}[1]{doi:\discretionary{}{}{}#1}\else
  \providecommand{\doi}{doi:\discretionary{}{}{}\begingroup
  \urlstyle{rm}\Url}\fi
\providecommand{\eprint}[2][]{\url{#2}}

\bibitem{mother}
{Abeyesinghe}, A., {Devetak}, I., {Hayden}, P., {Winter}, A.: The mother of all protocols: Restructuring quantum information＊s family tree.
\newblock Proc. Roy. Soc. London Ser. A, rspa20090202 (2009)

\bibitem{hastings}
{Aubrun}, G., {Szarek}, S., {Werner}, E.: Hastings's additivity counterexample via dvoretzky's theorem.
\newblock Comm. Math. Phys. \textbf{305}, 85--97 (2011)

\bibitem{baaj}
{Baaj}, S., {Skandalis}, G.: Unitaires multiplicatifs et dualit{\'e} pour les
  produits crois{\'e}s de mathrm $c^*$-alg{\`e}bres.
\newblock Ann. Sci. {\'E}cole Norm. Sup.,
  \textbf{26}, 425--488 (1993)

\bibitem{bell}
{Bell}, J.S.: On the Einstein-Podolsky-Rosen paradox.
\newblock Physics \textbf{1}, 195--200 (1964)

\bibitem{erasure}
{Bennett}, C.H., {DiVincenzo}, D.P., {Smolin}, J.A.: Capacities of quantum
  erasure channels.
\newblock Phys. Rev. Lett. \textbf{78}, 3217每3220 (1997)

\bibitem{BSST}
{Bennett}, C.H., {Shor}, P.W., {Smolin}, J., {Thapliyal}, A.V.:
  Entanglement-assisted capacity of a quantum channel and the reverse shannon
  theorem.
\newblock IEEE Trans. Inform. Theory \textbf{48}, 2637--2655 (2002)

\bibitem{BL}
{Bergh}, J., {L{\"o}fstr{\"o}m}, J.: Interpolation spaces. An introduction.
\newblock Berlin: Springer, 1976

\bibitem{blecher}
Blecher, D.P., Paulsen, V.I.: Tensor products of operator spaces.
\newblock J. Funct. Anal. \textbf{99}, 262--292 (1991)

\bibitem{bunkov}
{Bunkov}, Y., {Gazizulin}, R.: Majorana fermions: Direct observation in 3He.
\newblock arXiv:1504.01711

\bibitem{rep}
{Collins}, B., {\'{S}niady}, P.: Integration with respect to the Haar measure on unitary, orthogonal and symplectic group.
\newblock Comm. Math. Phys. \textbf{264}, 773-795 (2006)

\bibitem{Neufang}
{Crann}, J., {Neufang}, M.: Quantum channels arising from abstract harmonic
  analysis.
\newblock J. Phys. A \textbf{46} 045308 (2013)

\bibitem{cubitt}
{Cubitt}, T., {Elkouss}, D., {Matthews}, W., {Ozols}, M., {P{\'e}rez-Garc{i}a},
  D., {Strelchuk}, S.: Unbounded number of channel uses may be required to
  detect quantum capacity.
\newblock Nat. Commun. \textbf{6} (2015)

\bibitem{Devetak}
{Devetak}, I.: The private classical capacity and quantum capacity of a quantum
  channel.
\newblock IEEE Trans. Inform. Theory \textbf{51}, 44--55 (2005)

\bibitem{Family}
{Devetak}, I., {Harrow}, A.W., {Winter}, A.: A family of quantum protocols.
\newblock Phys. Rev. Lett. \textbf{93}, 230504 (2004)

\bibitem{DJKR}
{Devetak}, I., {Junge}, M., {King}, C., {Ruskai}, M.B.: Multiplicativity of
  completely bounded {$p$}-norms implies a new additivity result.
\newblock Comm. Math. Phys. \textbf{266}, 37--63 (2006).

\bibitem{DS}
{Devetak}, I., {Shor}, P.W.: The capacity of a quantum channel for simultaneous
  transmission of classical and quantum information.
\newblock Comm. Math. Phys. \textbf{256}, 287每303 (2005)

\bibitem{sa}
{DiVincenzo}, D.P., {Shor}, P.W., {Smolin}, J.A.: Quantum-channel capacity of
  very noisy channels.
\newblock Phys. Rev. A \textbf{57}, 830 (1998)

\bibitem{drinfeld}
{Drinfeld}, V.G.: Quantum groups
\newblock Zapiski Nauchnykh Seminarov {POMI} 155, 18-49 (1986)

\bibitem{dummit}
{Dummit}, D., {Foote}, R.M.: Abstract algebra. Hoboken: Wiley, 2004

\bibitem{ER}
{Effros}, E., {Ruan}, Z.: Operator spaces.
\newblock New York: Oxford University Press, 2000

\bibitem{EPR}
{Einstein}, A., {Podolsky}, B., {Rosen}, N.: Can quantum-mechanical description
  of physical reality be considered complete?
\newblock Physical Review \textbf{47}, 777每780 (1935)

\bibitem{kac}
{Enock}, M., {Schwartz}, J.M.: Kac algebras and duality of locally compact
  groups.
\newblock Berlin: Springer Science \& Business Media, 2013

\bibitem{determinant}
{Fuglede}, B., {Kadison}, R.V.: Determinant theory in finite factors.
\newblock Ann. of Math. 520--530 (1952)

\bibitem{Wolf}
{Fukuda}, M., {Wolf}, M.M.,: Simplifying additivity problems using direct sum constructions.
\newblock J. Math. Phys. \textbf{48}, 072101 (2007)

\bibitem{invcoh}
{Garcia-Patr\'{o}n}, R., {Pirandola}, S., {Lloyd}, S., {Shapiro}, J.H.: Reverse
  coherent information.
\newblock Phys. Rev. Lett. \textbf{102}, 210501 (2009)

\bibitem{adephasing}
{Giovannetti}, V., {Fazio}, R.: Information-capacity description of spin-chain
  correlations.
\newblock Phys. Rev. A \textbf{71}, 032314  (2005)

\bibitem{GW}
{Gupta}, M.K., {Wilde}, M.M.: Multiplicativity of completely bounded $p$-norms
  implies a strong converse for entanglement-assisted capacity.
\newblock Comm. Math. Phys. \textbf{334}, 867--887 (2015)

\bibitem{Haastandard}
{Haagerup}, U.: The standard form of von neumann algebras.
\newblock Mathematica Scandinavica \textbf{37}, 271--283 (1975)

\bibitem{haa}
{Haagerup}, U., {Musat}, M.: Factorization and dilation problems for completely
  positive maps on von neumann algebras.
\newblock Comm. Math. Phys. \textbf{303}, 555--594
  (2011)

\bibitem{hayden}
{Hayden}, P., {Winter}, A.: Counterexamples to the maximal p-norm
  multiplicativity conjecture for all $p> 1$.
\newblock Comm. Math. Phys. \textbf{284}, 263--280
  (2008)

\bibitem{Holevo}
{Holevo}, A.S.: The capacity of the quantum channel with general signal states.
\newblock IEEE Trans. Inform. Theory \textbf{44}, 269-273 (1998)

\bibitem{Holevo2}
{Holevo}, A.S., {Werner}, R,F,: Evaluating capacities of bosonic Gaussian channels.
\newblock Phys. Rev. A \textbf{63}, 032312 (2001)

\bibitem{JNR}
{Junge}, M., {Neufang}, M., {Ruan}, Z.: A representation theorem for locally
  compact quantum groups.
\newblock Int. J. Math. \textbf{20}, 377--400 (2009)

\bibitem{JNR2}
{Junge}, M., {Neufang}, M., {Ruan}, Z.: Reversed coherent information for
  quantum group channels.
\newblock Private communication (2009)

\bibitem{JP}
{Junge}, M., {Palazuelos}, C.: Cb-norm estimates for maps between
  noncommutative $ {L}_p $-spaces and quantum channel theory.
\newblock Internat. Math. Res. Notices, rnv161 (2015)

\bibitem{psumming}
{Junge}, M., {Palazuelos}, C.: Channel capacities via p-summing norms.
\newblock Adv. Math. \textbf{272}, 350--398 (2015)

\bibitem{kitaev-unpaired}
{Kitaev}, A.Y.: Unpaired majorana fermions in quantum wires.
\newblock Physics-Uspekhi \textbf{44}, 131 (2001)

\bibitem{kitaev}
{Kitaev}, A.Y.: Fault-tolerant quantum computation by anyons.
\newblock Ann. Physics \textbf{303}, 2--30 (2003)

\bibitem{li}
{Li}, J., {Chen}, H., {Drozdov}, I.K., {Yazdani}, A., {Bernevig}, B.A.,
  {MacDonald}, A.H.: Topological superconductivity induced by ferromagnetic
  metal chains.
\newblock Phys. Rev. B \textbf{90}, 235433 (2014)

\bibitem{Lloyd}
{Lloyd}, S.: Capacity of the noisy quantum channel.
\newblock Phys. Rev. A \textbf{55}, 1613 (1997)





\bibitem{renyi}
{M\"{u}ller-Lennert}, M., {Dupuis}, F., {Szehr}, O., {Fehr}, S., {Tomamichel},
  M.: On quantum R\'{e}nyi entropies: A new generalization and some properties.
\newblock J. Math. Phys. \textbf{54}, 122203 (2013)



\bibitem{nadj-perge}
{Nadj-Perge}, S., {Drozdov}, I., {Li}, J., et al: Observation of majorana
  fermions in ferromagnetic atomic chains on a superconductor.
\newblock  Science \textbf{346}, 602-607 (2014).

\bibitem{noteonentropy}
{Nakamura}, M., {Umegaki}, H.: A note on the entropy for operator algebras.
\newblock Proceedings of the Japan Academy \textbf{37}, 149--154 (1961)

\bibitem{pvp}
{Pisier}, G.: Noncommutative vector valued $ l_p $-spaces and completely $ p
  $-summing maps.
\newblock Ast\'e risque \textbf{247}, 1-131 (1993)

\bibitem{Psbook}
{Pisier}, G.: Introduction to operator space theory.
\newblock Cambridge University Press, 2003

%\bibitem{osqit}
%{Ruskai}, M.B., {Junge}, M., {Kribs}, D., {Hayden}, P., {Winter}, A.: Operator
  %structures in quantum information theory

\bibitem{SW}
{Schumacher}, B., {Westmoreland}, M.D.: Sending classical information via noisy
  quantum channels.
\newblock Phys. Rev. A \textbf{56}, 131 (1997)

\bibitem{48}
{Shannon}, C.E.: A mathematical theory of communication.
\newblock Bell System Technical Journal \textbf{27},379每423  (1948)

\bibitem{Shor}
{Shor}, P.W.: The quantum channel capacity and coherent information.
\newblock In lecture notes, MSRI Workshop on Quantum Computation. 2002

\bibitem{Winterss}
{Smith}, G., {Smolin}, J., {Winter}, A.: The quantum capacity
  with symmetric side channels.
\newblock IEEE Trans. Inform. Theory  \textbf{54}, 4208--4217
  (2008)

\bibitem{sa1}
{Smith}, G., {Smolin}, J.A.: Degenerate quantum codes for Pauli channels.
\newblock Phys. Rev. Lett. \textbf{98}, 030501 (2007)

\bibitem{approximate}
{Sutter}, D., {Scholz}, V.B., {Renner}, R.: Approximate degradable quantum channels.
\newblock 2015 IEEE International Symposium on Information Theory (ISIT) 2767--2771 (2015)

\bibitem{sc}
{Tomamichel}, M., {Wilde}, M.M., {Winter}, A.: Strong converse rates for quantum communication.
\newblock 2015 IEEE International Symposium on Information Theory (ISIT), 2386--2390 (2015)



\bibitem{Tak}
{Takesaki}, M.: Theory of operator algebras II.
\newblock Springer Science \& Business Media, 2013

\bibitem{vaes}
{Vaes}, S., {Vergnioux}, R.: The boundary of universal discrete quantum groups,
  exactness, and factoriality.
\newblock Duke Math. J. \textbf{140}, 35--84 (2007)

\bibitem{Pau}
{Vern}, P.: Completely bounded maps and operator algebras.
\newblock Cambridge University Press, 2003.


\bibitem{asym}
{Vershik}, A.M., {Kerov}, S.V.: Asymptotic of the largest and the typical
  dimensions of irreducible representations of a symmetric group.
\newblock Funct. Anal. Appl. \textbf{19}, 21--31 (1985)

\bibitem{XW}
{Wang}, X., {Duan}, R.: A semidefinite programming upper bound of quantum capacity.
\newblock arXiv:1601.06888

\bibitem{wilczek}
{Wilczek}, F.: Majorana returns.
\newblock Nat. Phys. \textbf{5}, 614--618 (2009)

\bibitem{Wildebook}
{Wilde}, M.M.: {Quantum information theory}.
\newblock  Cambridge University Press, 2013.

\bibitem{HW}
{Wilde}, M.M., {Hsieh}, M.: The quantum dynamic capacity formula of a quantum
  channel.
\newblock Quantum Information Process \textbf{11}, 1431--1463 (2012)

\bibitem{WWY}
{Wilde}, M.M., {Winter}, A., {Yang}, D.: Strong converse for the classical
  capacity of entanglement-breaking and {H}adamard channels via a sandwiched
  {R}enyi relative entropy.
\newblock Comm. Math. Phys. \textbf{331}, 593--622 (2014)

\bibitem{WD}
{Yang}, D., {Winter}, A.: Potential capacities of quantum channels (2015).
\newblock IEEE Trans. Inform. Theory  \textbf{62}, 1415--1424
  (2016)
\end{thebibliography}
